\documentclass[aps,prb,english,superscriptaddress,twocolumn,tightenlines,showpacs,longbibliography,floatfix,amsfonts]{revtex4-2}
\usepackage[T1]{fontenc}
\usepackage[utf8]{inputenc}
\usepackage[normalem]{ulem}
\usepackage{textcomp}
\usepackage{graphicx,float}
\usepackage{subfigure}
\usepackage[dvipsnames,svgnames]{xcolor}
\usepackage{float}
\usepackage{mathrsfs,mathtools,xfrac}
\usepackage{amsmath,amssymb,amsthm,amscd,bm,bbm}
\usepackage{txfonts} 
\usepackage{afterpackage,multirow,comment}
\usepackage{array,booktabs,minibox}
	\newcolumntype{P}[1]{>{\centering\arraybackslash}p{#1}} 
\usepackage[colorlinks,bookmarks=false,citecolor=blue,linkcolor=blue,urlcolor=blue,filecolor=blue]{hyperref}
\usepackage{wasysym}
\usepackage{listings}
\usepackage{framed}
\newcommand{\im}{\operatorname{im}}

\newcommand{\ket}[1]{| #1 \rangle}
\newcommand{\bra}[1]{\langle #1 |}

\newcommand{\tr}{\operatorname{Tr}}
\newcommand{\pt}{\partial}
\newcommand{\dg}{\dagger} 
\newcommand{\eps}{\varepsilon}

\newcommand{\gsd}{\operatorname{GSD}}

\newcommand{\ord}{\operatorname{ord}}
\renewcommand{\hom}{\operatorname{Hom}}

\newtheorem{thm}{Theorem}
\newtheorem{defn}{Definition}

\newtheorem{lem}{Lemma}
\newtheorem{prop}{Proposition}
\newtheorem{cor}{Corollary}
\newtheorem{fact}{Fact}

\begin{document}
\title{Symmetry-enriched topological order and quasifractonic behavior in $\mathbb{Z}_N$ stabilizer codes}
\author{Siyu He}
\affiliation{Institute of Theoretical Physics, Chinese Academy of Sciences, Beijing
100190, China}
\affiliation{School of Physical Sciences, University of Chinese Academy of Sciences,
Beijing 100049, China.}
\author{Hao Song}
\email{songhao@itp.ac.cn}

\affiliation{Institute of Theoretical Physics, Chinese Academy of Sciences, Beijing
100190, China}

\begin{abstract}
We study a broad class of qudit stabilizer codes, termed $\mathbb{Z}_N$ bivariate-bicycle (BB) codes, arising either as two-dimensional realizations of modulated gauge theories or as $\mathbb{Z}_N$ generalizations of binary BB codes. Our central finding, derived from the polynomial representation, is that the essential topological properties of these $\mathbb{Z}_N$ codes can be determined by the properties of their $\mathbb{Z}_p$ counterparts, where $p$ are the prime factors of $N$, even when $N$ contains prime powers ($N = \prod_i p_i^{k_i}$). This result yields a significant simplification by leveraging the well-studied framework of codes with prime qudit dimensions. In particular, this insight directly enables the generalization of the algebraic-geometric methods (e.g., the Bernstein-Khovanskii-Kushnirenko theorem) to determine anyon fusion rules in the general qudit situation. Moreover, we elucidate the symmetry-enriched topological (SET) order underlying the quasifractonic behavior in qudit BB codes (including the Delfino-Chamon-You model), resolving the associated anyon mobility puzzle. We also develop an efficient computational algebraic method, based on Gr\"{o}bner bases over the ring of integers, to determine both the topological order and its SET properties.
\end{abstract}

\maketitle

\section{Introduction}

Topological order offers a unifying framework for understanding exotic phases of quantum matter indistinguishable by local order parameters. Characterized by long-range entanglement, emergent gauge fields, robust ground state degeneracy, and anyonic excitations with nontrivial braiding statistics~\cite{Wen1990, Wen1990_2, Kitaev2003, Kitaev2006, Levin2006, Kitaev2006_2, Wen2017}, it also offers promising proposals for fault-tolerant quantum computation by encoding quantum information in nonlocal degrees of freedom~\cite{Dennis2002, Nayak2008}.
  
The toric code furnishes a simple and illustrative example of a physical system exhibiting topological order and provides a paradigmatic model for fault-tolerant quantum computation~\cite{Kitaev2003}. Recently, there has been increasing interest in generalizing the toric code in both the condensed matter and the quantum computing communities, with each following distinct yet complementary paths.
  
In quantum computing, one emphasis has been on quantum low-density parity-check (QLDPC) codes~\cite{Postol2001, MacKay2004, Kovalev2013, Bravyi2014, Tillich2014, Gottesman2014, Babar2015, Hastings2021, Breuckmann2021, Panteleev2021, Panteleev2022-1,  Panteleev2022-2, Breuckmann2021-2, Tremblay2022, Raveendran2022, Xu2024}, which pave the way to improve the scalability limits and mitigate the qubit overhead problem of the toric code. Among them, the bivariate-bicycle (BB) codes~\cite{Bravyi2024}, a natural generalization of the toric code proposed by IBM, have garnered notable attention~\cite{Bravyi2022, Shaw2025, Eberhardt2024, Eberhardt2025, Berthusen2025, Cowtan2024, Gong2024, Poole2024, Williamson2024, Wang2024, Liang2024-2, Blue2025, Sayginel2025, Yoder2025, Postema2025, Voss2025}, as they promise to sharply reduce the qubit overhead without requiring significant long-range resources. Recently, significant progress has been made in understanding both classical and quantum LDPC codes from the condensed matter perspective~\cite{Lim2021, Rakovszky2023, Rakovszky2024, Deroeck2024, Liang2024-1, Yin2025, Placke2024}. Specifically, the anyon theory of BB codes has been established~\cite{Chen2025, Liang2025}, revealing close connections between the quantum error correction properties of BB codes and their topological order.
  
From the condensed matter perspective, generalizations of the toric code have been motivated by exploring symmetry structures and alternative phases of matter. A natural starting point is to generalize binary systems to $\mathbb{Z}_{N}$ systems describing more general Abelian symmetries. In particular, modulated symmetries are global symmetries dependent on the spatial details of the lattice~\cite{Sala2022, Watanabe2023, Delfino2023_2, Delfino2024, Sala2024, Han2024, Hiromi2024_1, Hiromi2024_2, Pace2025, Yoshitome2025, Kim2025_1, Kim2025_2}. Exploring modulated symmetries leads to the discovery of rich phenomena, including Hilbert space fragmentation~\cite{Sala2020, Khemani2020, Moudgalya2020} and UV/IR mixing~\cite{Oh2022_1, Oh2022_2, Gorantla2021, Gorantla2022, Pace2022, Casasola2024, Kim2025_3}. A remarkable example of modulated symmetries is the subsystem symmetry, which gives rise to fractons that are immobile in isolation~\cite{Chamon2005, Bravyi2011, Haah2011, Yoshida2013, Vijay2015, Vijay2016, Williamson2016, Pretko2017, Ma2017, Shirley2019, Nandkishore2019, Pretko2020, Song2019, Prem2019, Devakul2019, Song2022, Song2024, Canossa2024, Zhang2025}. These insights from modulated symmetries have inspired further generalizations of the toric code~\cite{Watanabe2023, Delfino2023_2, Delfino2024}, uncovering unconventional properties beyond conventional topological order~\cite{Delfino2023, Chen2025}. These generalizations share many similarities with those considered in the quantum computing community. We expect similar methods to be applicable to both studies. However, in contrast to quantum computing, discussions beyond the binary situation present additional challenges. 
  
In this work, we consider various generalizations of toric codes for $\mathbb{Z}_N$ qudits exhibiting modulated symmetries. These models fall into a broad class of stabilizer codes we construct, which we refer to as $\mathbb{Z}_N$ BB codes. We employ the polynomial representation~\cite{Haah2013} to describe their topological orders. Notably, we find that the essential topological-order properties of $\mathbb{Z}_N$ qudit codes—including the topological condition and the fusion rules—are determined once these properties are known for the corresponding $\mathbb{Z}_p$ codes, where $p$ runs over the prime factors of $N$. This is a significant simplification, as the topological orders of stabilizer codes with prime qudit dimensions have been well studied. Specifically, it enables us to generalize the algebraic-geometric method of Ref.~\cite{Chen2025}, which uses the Bernstein-Khovanskii-Kushnirenko (BKK) theorem to determine fusion rules, from qubits to the general qudit cases. In addition, we elucidate the symmetry-enriched topological (SET) structure~\cite{Essin2013, Mesaros2013, Huang2014, Chiu2016, Barkeshli2019} underlying the quasifractonic behavior. We develop an efficient computational algebraic method, based on the Gr\"{o}bner bases over integers, to determine both the topological order and its symmetry-enriched properties for qudit BB codes. As an application, our framework provides a precise characterization of the Delfino-Chamon-You (DCY) model~\cite{Delfino2023_2} and resolves its mobility puzzle regarding whether one-to-one anyon hopping is permitted over any finite distance.
  
This paper is structured as follows. In Sec.~\ref{sec2}, we introduce the models considered from a gauging perspective, followed by a review of the polynomial representation of stabilizer codes in Sec.~\ref{sec3}. We then analyze BB codes, discussing their topological orders in Sec .~\ref {sec4} and their quasifractonic behaviors in Sec.~\ref{sec5}. Next, Sec.~\ref{sec6} examines the topological orders and quasifractonic behaviors of the models introduced in Sec.~\ref{sec2}. The computational algebraic method used to obtain these results is presented in Sec.~\ref{sec7}. Finally, we summarize and discuss our findings in Sec.~\ref{sec8}.

\section{Models from gauging symmetries \label{sec2}}

Various $\mathbb{Z}_N$ generalizations of the toric code are considered in this paper, where $N$ can be any positive integer. These models are all defined on a square lattice with a $\mathbb{Z}_N$ qudit assigned to each edge. 

For each qudit, we introduce the generalized Pauli matrices
\begin{equation}
    X=\sum_{j\in\mathbb{Z}_N}\ket{j}\bra{j+1},\quad\quad Z=\sum_{j\in\mathbb{Z}_N}\omega^{j}\ket{j}\bra{j},\label{eq:definition_of_Pauli_operators}
\end{equation}
where $\omega=e^{\frac{2\pi i}{N}}$. For brevity, we refer to generalized Pauli simply as Pauli below. They satisfy $X^N=Z^N=1$ and the commutation relation $XZ=\omega ZX$. 

Qudits are assumed to be located at the centers of edges. Accordingly, as depicted in Fig.~\ref{fig:stabilizers_of_generalized_toric_code}(a) and (b), we denote the Pauli operators for qudits on horizontal edges by $X_{(i+1/2, j)}$ and $Z_{(i+1/2, j)}$, and those on vertical edges by $X_{(i, j+1/2)}$ and $Z_{(i, j+1/2)}$.

\begin{figure}[t]
    \centering
    \includegraphics[width=\linewidth]{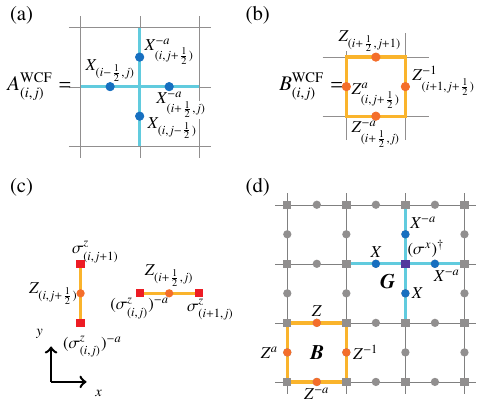}
    \caption{Illustration of the WCF model. (a) The $X$-type stabilizer $A_{(i,j)}^{\operatorname{WCF}}$. (b) The $Z$-type stabilizer $B_{(i,j)}^{\operatorname{WCF}}$. (c) The minimal coupling terms between gauge and matter qudits. (d) The zero flux terms (yellow) and the gauging symmetry operators (light blue). The dark blue and the orange dots correspond to the gauge-qudit Pauli $X$ and $Z$ operators on the edges, respectively. Similarly, the red and purple squares indicate the matter-qudit operators 
    $\sigma_x$ and $\sigma_z$ at the vertices. The explicit forms of each operator are shown next to their symbols. We omit the coordinates of operators in (d) for brevity.} 
    \label{fig:stabilizers_of_generalized_toric_code}
\end{figure}

We focus on stabilizer models of the Calderbank-Shor-Steane type, each defined by a Hamiltonian of the form
\begin{equation}
    H=-\sum_{i,j}A_{(i,j)}-\sum_{i,j}B_{(i,j)}+\operatorname{H.c.},\label{eq:Hamiltonian}
\end{equation}
where $A_{(i,j)}$ and $B_{(i,j)}$ are the products of Pauli $X$ and $Z$ operators, respectively, constructed such that all $A$ and $B$ operators mutually commute. Lattice translation symmetry is respected, with $(i,j)$ labeling the positions of $A$ and $B$ operators.

\subsection{Watanabe-Cheng-Fuji model\label{sec2a}}
The first model under investigation is the Watanabe-Cheng-Fuji (WCF) model~\cite{Watanabe2023}
\begin{equation}
H_{\mathrm{WCF}}=-\sum_{i,j}A_{\left(i,j\right)}^{\mathrm{WCF}}-\sum_{i,j}B_{\left(i,j\right)}^{\mathrm{WCF}}+\mathrm{H.c.}, \label{eq:WCF}
\end{equation}
whose stabilizers take the form
\begin{equation}
    \begin{aligned}
    A_{(i,j)}^{\operatorname{WCF}}&= X^{-a}_{(i+\frac{1}{2},j)}X^{-a}_{(i,j+\frac{1}{2})}X_{(i-\frac{1}{2},j)}X_{(i,j-\frac{1}{2})}, \\
    B_{(i,j)}^{\operatorname{WCF}}&= Z_{(i+1,j+\frac{1}{2})}^{-1}Z_{(i+\frac{1}{2},j+1)}Z^{-a}_{(i+\frac{1}{2},j)}Z^{a}_{(i,j+\frac{1}{2})},
\end{aligned}\label{eq:stabilizers of generalized toric code}
\end{equation}
as shown in Fig.~\ref{fig:stabilizers_of_generalized_toric_code}(a) and (b), where $a\in \mathbb{Z}_N$. 

For the case where $a$ is coprime to $N$, the WCF model can be naturally interpreted as a generalized lattice gauge theory, emerging from the process of gauging an exponential symmetry. Explicitly, we follow the procedure described in Ref.~\cite{Shirley2019}. We introduce $\mathbb{Z}_N$ qudits (as the matter field) on the vertices of the lattice, governed by the trivial paramagnetic Hamiltonian
\begin{equation}
    H_0=-\sum_{i,j}\left(\sigma_{(i,j)}^{x}+\mathrm{H.c.}\right),\label{eq:Hamiltonian_of_matter}
\end{equation}
as a system satisfying the exponential symmetry
\begin{equation}
    U_{\operatorname{exp}}=\prod_{i,j}(\sigma_{(i,j)}^{x})^{a^{i+j}},\label{eq:symmetry_GTC}
\end{equation}
where $\sigma_{(i,j)}^{x}$ is the generalized Pauli $X$ operator associated with the matter qudit, and it acts with an exponential spatial dependence $a^{i+j}$~\cite{note}. 
This exponential symmetry exemplifies a spatially modulated global symmetry, which generalizes the conventional uniform case.
To gauge the symmetry, we introduce the gauge field $Z_{(i+\frac{1}{2},j)}$ and $Z_{(i,j+\frac{1}{2})}$ on lattice edges, coupled to symmetric operators $(\sigma^{z}_{(i,j)})^{-a} \sigma^{z}_{(i+1,j)}$ and $(\sigma^{z}_{(i,j)})^{-a}\sigma^{z}_{(i,j+1)}$, as in Fig.~\ref{fig:stabilizers_of_generalized_toric_code}(c). The Hilbert space is constrained by the generalized Gauss law $G_{(i,j)}^{\operatorname{WCF}}=1$, where
\begin{equation}
    G_{(i,j)}^{\operatorname{WCF}}=(\sigma_{(i,j)}^{x})^{\dagger}A_{(i,j)}^{\operatorname{WCF}}.\label{eq:GTC_local_symmetry}
\end{equation}
Thus, the matter Hamiltonian $H_0$ in Eq.~\eqref{eq:Hamiltonian_of_matter} can be rewritten as $H_0 = -\sum (A^{\operatorname{WCF}}_{(i,j)} + \operatorname{h.c.})$. Additionally, gauge fluxes remain to be gapped out by adding terms $B_{(i,j)}^{\operatorname{WCF}}$ and $(B_{(i,j)}^{\operatorname{WCF}})^{\dagger}$, constructed as products of gauge field Pauli $Z$ operators that respect the Gauss law constraint. 
We thus obtain a gauge theory governed by $H_{\operatorname{WCF}}$. It can be further reduced into a spin model by the unitary transformation $\prod_{v} CA^{\operatorname{WCF}}_{v}$, 
which leaves $H_{\operatorname{WCF}}$ invariant and trivializes the Gauss law constraint $G_{v}^{\operatorname{WCF}}=1$ to $\sigma_{v}^{x}=1$, where 
\begin{equation}
CA_{v}=\sum_{m\in\mathbb{Z}_{N}}P_{v,m}A_{v}^{m}
\end{equation}
with $P_{v,m}$ denoting the projector selecting $\sigma_v^{z}=\omega^m$ at vertex $v=(i,j)$. This results in exactly the spin model described by Eq.~\eqref{eq:WCF}.

\subsection{Delfino-Chamon-You model\label{sec2b}}
Another intriguing example related to exponential symmetry is the DCY model~\cite{Delfino2023_2}, whose excitation mobility still remains a puzzle, and will be addressed in this paper. Its Hamiltonian is given by
\begin{equation}
    H_{\operatorname{DCY}}=-\sum_{i,j}A_{(i,j)}^{\operatorname{DCY}}-\sum_{i,j}B_{(i,j)}^{\operatorname{DCY}}+\operatorname{h.c.},
\end{equation}
with terms depicted in Fig.~\ref{fig:exponential symmetry model}(a) and (b). Explicitly, 
\begin{equation}
    \begin{aligned}
    A_{(i,j)}^{\operatorname{DCY}}&= X^{a}_{(i+\frac{1}{2},j)}X^{a}_{(i,j+\frac{1}{2})}X_{(i-\frac{1}{2},j)}^{-(a+1)}X_{(i,j-\frac{1}{2})}^{-(a+1)}X_{(i-\frac{3}{2},j)}X_{(i,j-\frac{3}{2})}, \\
    B_{(i,j)}^{\operatorname{DCY}}&= Z_{(i+\frac{1}{2},j)}^{a} Z_{(i,j+\frac{1}{2})}^{-a}
    Z_{(i+\frac{1}{2},j+1)}^{-(a+1)}Z_{(i+1,j+\frac{1}{2})}^{a+1}
    Z_{(i+\frac{1}{2},j+2)}Z^{-1}_{(i+2,j+\frac{1}{2})}  
\end{aligned}\label{eq:stabilizers of exponential symmetry model},
\end{equation}
where $a\in \mathbb{Z}_N$. Note that there are two qudits per unit cell, which we take to be located at $(i-\frac{1}{2},j)$ and $(i,j-\frac{1}{2})$; by shifting them to the lattice vertices, the model may be transformed into precisely the form presented in Ref.~\cite{Delfino2023_2}.

\begin{figure}[t]
    \centering
    \includegraphics[width=\linewidth]{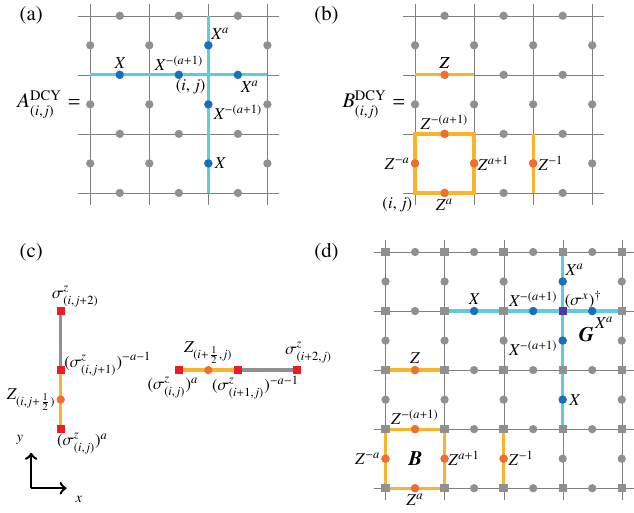}
    \caption{Illustration of the DCY model. (a) $X$-type stabilizer $A_{(i,j)}^{\operatorname{DCY}}$.
    (b) $Z$-type stabilizer $B_{(i,j)}^{\operatorname{DCY}}$. 
    In (a) and (b), the location of the stabilizers is marked by the coordinate $(i,j)$ in each figure, while their constituent Pauli operators are shown without explicit coordinates. 
    (c) The minimal coupling terms between gauge and matter qudits. 
    (d) The zero flux condition (yellow) and the gauging symmetry operators (light blue) with all coordinate labels omitted. 
    Throughout the figures, deep blue dots, orange dots, red squares, and purple squares represent the $X$, $Z$, $\sigma^z$, and $\sigma^x$ operators, respectively.}
    \label{fig:exponential symmetry model}
\end{figure}

The DCY model can also be obtained by the gauging procedure. We start with the same matter Hamiltonian $H_0$ in Eq.~\eqref{eq:Hamiltonian_of_matter}, but consider a different symmetry structure generated by
\begin{equation}
\begin{array}{rlrl}
U_0 &= \prod_{i,j} \sigma_{(i,j)}^x, & \quad
U_x &= \prod_{i,j} (\sigma_{(i,j)}^x)^{a^i}, \\
U_y &= \prod_{i,j} (\sigma_{(i,j)}^x)^{a^j}, & \quad
U_{\mathrm{exp}} &= \prod_{i,j} (\sigma_{(i,j)}^x)^{a^{i+j}},
\end{array} \label{eq:sym_DCY}
\end{equation}
including the uniform symmetry $U_0$,  the exponential symmetry $U_{\operatorname{exp}}$ discussed previously, and two additional symmetries with exponential modulation along one spatial coordinate each~\cite{note}.
If these global symmetries are independent, 
symmetric operators are generated by $(\sigma_{(i,j)}^{z})^{a}(\sigma_{(i+1,j)}^{z})^{-(a+1)}\sigma_{(i+2,j)}^{z}$ and $(\sigma_{(i,j)}^{z})^{a}(\sigma_{(i,j+1)}^{z})^{-(a+1)}\sigma_{(i,j+2)}^{z}$, to which gauge field qudits $Z_{(i+\frac{1}{2},j)}$ and $Z_{(i,j+\frac{1}{2})}$ are associated. 
Applying the gauge procedure results in the DCY model. See Fig.~\ref{fig:exponential symmetry model}(c) and (d).

\subsection{Qudit BB codes and hypergraph product codes\label{sec2c}}

More generally, we may construct a broad family of stabilizer Hamiltonians of the form Eq.~\eqref{eq:Hamiltonian}, with
\begin{equation}
    \begin{aligned}
    A_{(i,j)}&= \prod_{i',j'}X^{c_{i',j'}}_{(i-i'+\frac{1}{2},j-j')}X^{d_{i',j'}}_{(i-i',j-j'+\frac{1}{2})}, \\
    B_{(i,j)}&= \prod_{i',j'}Z^{d_{i',j'}}_{(i+i'+\frac{1}{2},j+j')}Z^{-c_{i',j'}}_{(i+i',j+j'+\frac{1}{2})}, 
\end{aligned}\label{eq:stabilizers_of_BB_code_model}
\end{equation}
where $c_{i',j'}$ and $d_{i',j'}\in\mathbb{Z}_N$ specify the powers of Pauli matrices on the lattice edges. We assume only finitely many of the coefficients $c_{i',j'}$ and $d_{i',j'}$ are nonzero, ensuring that each stabilizer is locally supported on the lattice. We refer to them as \textit{qudit BB codes}, which generalize the qubit BB codes introduced in Ref.~\cite{Bravyi2024} to the $\mathbb{Z}_N$ setting.
  
If $c_{i',j'}=0$ for $j'\neq 0$ and $d_{i',j'}=0$ for $i'\neq 0$, we refer to this type of BB code as a \textit{hypergraph product (HGP) code}, since it arises from two classical codes via a tensor product construction in quantum error correction~\cite{Tillich2014}. Explicitly, the stabilizers of HGP codes take the form
\begin{equation}
    \begin{aligned}
    A_{(i,j)}&= \prod_{i',j'}X^{c_{i',0}}_{(i-i'+\frac{1}{2},j)}X^{d_{0,j'}}_{(i,j-j'+\frac{1}{2})}, \\
    B_{(i,j)}&= \prod_{i',j'}Z^{d_{0,j'}}_{(i+i'+\frac{1}{2},j)}Z^{-c_{i',0}}_{(i,j+j'+\frac{1}{2})}, 
\end{aligned}
\end{equation}
  
Both \( H_{\operatorname{WCF}} \) and \( H_{\operatorname{DCY}} \) are special cases of this family. The WCF model is obtained by setting \( c_{00} = d_{00} = -a \) and \( c_{10} = d_{01} = 1 \), with all other $c_{i',j'}$ and $d_{i',j'}$ set to zero. The DCY model, on the other hand, corresponds to \( c_{00} = d_{00} = a \), \( c_{10} = d_{01} = -(a + 1) \), and \( c_{20} = d_{02} = 1 \), with all remaining parameters vanishing.  

Commutation of the $A$ and $B$ terms can be explicitly verified using the polynomial formalism presented below. 
 
In addition, to enable a systematic investigation of this family of models and resolve remaining puzzles in the DCY model, we will develop efficient methods that generalize the algebraic techniques originally devised for binary codes~\cite{Haah2013,Chen2025}.

\section{Polynomial representation\label{sec3}}
The polynomial representation~\cite{Haah2013} provides a convenient and powerful tool for describing and studying stabilizer codes with translation symmetry. Here, we outline the representation adapted to $\mathbb{Z}_N$ qudit systems.

The polynomial ring formalism is set up through the Laurent polynomial ring $R=\mathbb{Z}_N[x^\pm,y^\pm]$, which can be viewed as the group ring $\mathbb{Z}_N[\Lambda]$. Here, $\Lambda=\{x^iy^j|i,j\in \mathbb{Z}\}$ denotes the lattice translation group, with $(i,j)$ locating unit cells. Here, we choose the convention that the unit cell contains two qudits at $(i+\frac{1}{2},j)$ and $(i,j+\frac{1}{2})$. Within this framework,  Pauli operators are represented as vectors in $R^4$, in particular,
\begin{equation}
    \begin{aligned}
    &X_{(i+\frac{1}{2},j)}^{c}=\begin{pmatrix}
        cx^iy^j\\0\\0\\0
    \end{pmatrix},\quad X_{(i,j+\frac{1}{2})}^{c'}=\begin{pmatrix}
        0\\c'x^iy^j\\0\\0
    \end{pmatrix},\\
    &Z_{(i+\frac{1}{2},j)}^{d}=\begin{pmatrix}
        0\\0\\dx^iy^j\\0
    \end{pmatrix},\quad Z_{(i,j+\frac{1}{2})}^{d'}=\begin{pmatrix}
        0\\0\\0\\d'x^iy^j
    \end{pmatrix},
    \end{aligned}\label{eq:polynomial_representations}
\end{equation}
where coefficients $c$, $c'$, $d$ and $d'\in \mathbb{Z}_N$ specify the powers of the corresponding Pauli matrices. 
Multi-qudit Pauli operators $\prod_{\ell}X_{\ell}^{c_\ell} \cdot \prod_{\ell'}Z_{\ell'}^{d_{\ell'}}$ are represented by summing the vectors of their constituent factors.

For two multi-qudit Pauli operators represented as $v_1,v_2\in R^4$, they commute with each other if and only if the constant term of the symplectic bilinear form ${v}_1^\dg\lambda_qv_2$ vanishes. Here, $v_1^{\dg}=\bar{v}_1^T$ where $T$ denotes the transpose operation and $\bar{v}_1(x,y)=v_1(\bar{x},\bar{y})$ with $\bar{x}=x^{-1}$ and $\bar{y}=y^{-1}$, and $\lambda_q$ denotes the standard symplectic form matrix, given by
\begin{equation}
    \lambda_q=\begin{pmatrix}
        0&\mathbf{1}_{2\times 2}\\-\mathbf{1}_{2\times 2}&0
    \end{pmatrix}.
\end{equation}
  
In the polynomial representation, the $\mathbb{Z}_N$ qudit BB code in Sec.~\ref{sec2c} is specified by two Laurent polynomials $f=\sum_{i,j}c_{i,j}x^iy^j$ and $g=\sum_{i,j}d_{i,j}x^iy^j$. The stabilizer group of the BB code is described by the image of the stabilizer map $\partial$, which is defined as 
\begin{align}
    \partial= \begin{pmatrix}
        \bar{f}&0 \\
        \bar{g}&0 \\
        0&g\\
        0&-f
        \end{pmatrix}=\begin{pmatrix}
            \partial_A&0 \\
            0&\partial_B
            \end{pmatrix}
        \label{eq:pial}.
\end{align}
In particular, the first and second columns of $\partial$ can generate all $X$-type stabilizers $A_{(i,j)}$ and $Z$-type stabilizers $B_{(i,j)}$ by translation actions, respectively. The symplectic bilinear form of the first and second columns of $\partial$ equals $0$, so all stabilizers of BB codes commute with each other, and BB codes are well defined.
  
Moreover, all excitation patterns created by local operators of BB codes are described by the image of an excitation map $\eps$. As in Ref.~\cite{Haah2013}, the excitation map of the BB code specified by Laurent polynomials $f$ and $g$ is expressed as
\begin{align}
    \eps=\begin{pmatrix}
        0&0&f&g\\
        -\bar{g}&\bar{f}&0&0
        \end{pmatrix}=\begin{pmatrix}
            0&\eps_A\\
            \eps_B&0
            \end{pmatrix}.
        \label{eq:eps}
\end{align}
All stabilizers commute with the Hamiltonian, ensuring that $\eps \circ\partial = 0$. Formally, this means that the stabilizer map $\partial$ and the excitation map $\eps$ form the following chain complex:  
\begin{equation}
    R^2\xrightarrow{\partial}R^4\xrightarrow{\eps}R^2.\label{eq:cc}
\end{equation}

For convenience, we introduce the following notation for BB codes.
\begin{defn}
    The $\mathbb{Z}_N$ qudit BB code specified by Laurent polynomials $f$ and $g\in R$ is referred to as the $\mathbf{BB}_N(f,g)$ code.
\end{defn}

Finally, we give the explicit polynomials specifying the WCF model and the DCY model as concrete examples. Namely, the WCF model in Sec.~\ref{sec2a} is specified by
\begin{equation}
    f=x-a,\quad g=y-a,\label{eq:poly1}
\end{equation} 
while the DCY model in Sec.~\ref{sec2b} is specified by
\begin{equation}
    f=x^2-(a+1)x+a,\quad g=y^2-(a+1)y+a.\label{eq:poly2}
\end{equation}

\section{Topological order of BB codes\label{sec4}}
In this section, we employ the polynomial representation to characterize the topological order of BB codes. Notably, we prove that key topological features, including the topological condition and the anyon fusion rules of the $\mathbb{Z}_N$ qudit BB code, can be systematically obtained through the $\mathbb{Z}_p$ qudit counterparts of this code for all prime factors $p$ of $N$.
\subsection{Topological condition\label{sec4a}}
A gapped quantum many-body Hamiltonian is topological when it exhibits topological order, i.e., its degenerate ground states cannot be distinguished by any local operator in the thermodynamic limit. In stabilizer code models like BB codes, this requires that any operator supported on a finite region of an infinite lattice that commutes with every stabilizer must, up to a phase, belong to the stabilizer group.
  
In the polynomial representation, the stabilizer group of BB codes is characterized by $\im\partial$. Since $\eps$ is the excitation map, the kernel of $\eps$, denoted as $\ker\eps$, represents all finite-support operators that cannot create excitations, and hence commute with all stabilizers. Therefore, BB codes exhibit topological order exactly when
\begin{equation}
    \ker\eps=\im\partial.\label{eq:topo_condition}
\end{equation}
In other words, the chain complex Eq.~\eqref{eq:cc} is exact. Given the explicit form of $\partial$ and $\eps$ in Eqs.~\eqref{eq:pial}~and~\eqref{eq:eps}, respectively, the $\mathbf{BB}_N(f,g)$ code is topological if and only if $\ker\eps_A=\im\partial_B$, or equivalently, $\ker\eps_B=\im\partial_A$.
  
As described by the below theorem, the topological condition of the $\mathbf{BB}_{N}(f,g)$ code can be checked by considering the $\mathbb{Z}_p$ counterparts of this code for all $p\mid N$, namely $\mathbf{BB}_{p}(f,g)$ code, whose polynomials is defined over $R_p=\mathbb{Z}_p[x^{\pm},y^{\pm}]\cong R/pR$.
\begin{thm}
    The $\mathbf{BB}_N(f,g)$ code is topological if and only if for all prime $p\mid N$, the $\mathbf{BB}_p(f,g)$ code is topological.\label{lem:topo_prime_factors}
\end{thm}
We present the detailed proof in Appendix~\ref{app1}. Applying Theorem~\ref{lem:topo_prime_factors} leads to a convenient criterion for topological order, presented in the following corollary.
\begin{cor}
    The $\mathbf{BB}_N(f,g)$ code is topological if and only if the quotient module
    \begin{equation}
        \mathcal{Q}=\frac{R}{(f,g)}
    \end{equation}
    is finite, where $(f,g)$ denotes the ideal generated by polynomials $f$ and $g$.\label{cor:topo_finite}
\end{cor}
The proof of this criterion is offered in Appendix~\ref{app1}. As we discuss in the following section, if the BB code is topological, $\mathcal{Q}$ describes the anyon types of the electric charge sector of this model, implying the total number of anyon types is finite. This result aligns with the general mobility theorem for topological excitations of translationally invariant 2D stabilizer codes in Ref.~\cite{Ruba2024}.
\subsection{Anyon types and fusion rules \label{sec4b}}
In a topologically ordered phase, a fundamental problem is to classify its anyonic excitation patterns. Two anyonic excitations belong to the same type if they can be transformed into each other by local operators, and the fusion rules specify how combining two anyons can yield a composite excitation with a possibly different anyon type. In this subsection, we investigate how to obtain the anyon types and fusion rules of the BB code via its polynomial representation.
  
As we show in Appendix~\ref{app2}, if the BB code is topological, every element in $R^2$ specifies a legitimate excitation pattern,  which means it can be created by either a local or an infinitely extended operator. Hence, for the $\mathbf{BB}_N(f,g)$ code, all anyon types and their fusion rules are specified by the quotient module
\begin{equation}
    \mathcal{C}=\frac{R^{2}}{\mathrm{im}\eps}=\mathcal{Q}\oplus\overline{\mathcal{Q}},\label{eq:anyon}
\end{equation}
with $\mathcal{Q}$ and $\overline{\mathcal{Q}}$ defined as:
\begin{equation}
    \mathcal{Q}=\frac{R}{\left(f,g\right)},\qquad\overline{\mathcal{Q}}=\frac{R}{(\bar{f},\bar{g})},
\end{equation}
since $\im\eps$ captures precisely those excitation patterns that are realizable by local operators, and two elements $\omega,\omega'\in R^2$ represent two excitation patterns that can transform into each other by local operators if and only if they are equivalent in $\mathcal{C}$. The sectors $\mathcal{Q}$ and $\overline{\mathcal{Q}}$ in Eq.~\eqref{eq:anyon} specify the $X$ and $Z$-type excitations of BB codes. Physically, they correspond to the electric and magnetic charges of the model in the emergent gauge theory, respectively.
  
Like the topological condition, the anyon types and fusion rules of the $\mathbf{BB}_N(f,g)$ code are also determined by those of the $\mathbf{BB}_p(f,g)$ codes for all prime factors $p$ of $N$. Explicitly,

\begin{thm}
Let $N=\prod_{i=1}^n p_i^{k_i}$ be the prime factorization of $N$.
Then, the fusion rules of anyons in the $\mathbf{BB}_N(f,g)$ code, which correspond to the additive group structure of $\mathcal{C}$, is given by
\begin{equation}
\mathcal{C} \cong \bigoplus_{i=1}^{n} \mathbb{Z}_{p_i^{k_i}}^{2Q_i}, \label{eq:anyon_relation}
\end{equation}
where, for each $i$, the integer $Q_i$ is determined by
\begin{equation}
Q_i = \dim_{\mathbb{Z}_{p_i}}\left(\frac{R_{p_i}}{(f,g)}\right)
= \frac{1}{2}\dim_{\mathbb{Z}_{p_i}}(\mathcal{C}_{p_i}),
\end{equation}
and $\mathcal{C}_{p_i}$ denotes the anyon module of the $\mathbf{BB}_{p_i}(f,g)$ code.
\label{thm:anyon_type_relation}
\end{thm}

In what follows, we refer to $Q_i$ as the \emph{topological index}. The proof of Theorem~\ref{thm:anyon_type_relation} involves basic homological algebra and is detailed in Appendix~\ref{app2.5}. Theorem~\ref{thm:anyon_type_relation} provides a convenient and efficient way to evaluate the anyon types and their fusion rules of BB codes using results from prime qudit dimensions. Specifically, this enables us to extend the algebraic-geometric approach for determining fusion rules based on the BKK theorem.  This method was previously limited to qubit BB codes~\cite{Chen2025}; it can now be applied to composite qudit dimensions. See Secs.~\ref{sec6c} and \ref{sec6d} for two concrete applications of Theorem~\ref{thm:anyon_type_relation}.

\section{Quasifractonic behavior and SET orders of BB Codes\label{sec5}}
In this section, we analyze anyon hopping and the GSD, revealing how their distinctive features emerge from the interplay between topological order and translation symmetry.

\subsection{Quasifractonic mobility of anyons}

\begin{figure}[t]
    \centering
    \includegraphics[width=\linewidth]{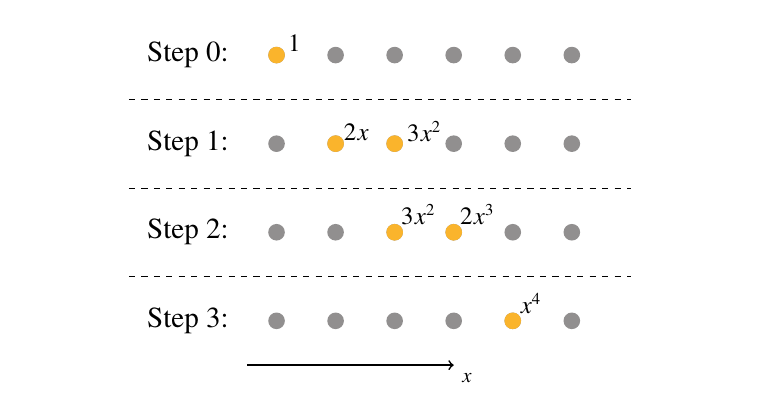}
    \caption{The quasifractonic behavior of the DCY model. Here $N=4$ and $a=1$. This figure depicts the process of moving excitations in the DCY model. The gray dots represent $Z$-type stabilizers at each position, where we only draw the stabilizers along the $\hat{x}$ directions in each step. The yellow dots represent the excitations, denoted by the monomials next to the dots. In step 1, the anyonic excitations are split by moving. After step 3, the original excitation is moved from $(0,0)$ to $(4,0)$.}
    \label{fig:quasi_fracton}
\end{figure}

Like the qubit BB codes~\cite{Chen2025}, anyons in $\mathbb{Z}_N$ BB codes typically cannot be moved smoothly by local operators.
  
To illustrate this, we consider a concrete example: the DCY model, and for simplicity assume that $a$ and $N$ are coprime [see Eq.~\eqref{eq:poly2}]. A complete treatment of anyon mobility in this model is presented in Sec.~\ref{sec6b}.

We focus on $Z$-type excitations, which can be labeled as an element in $R$. Initially, assume that the lowest energy excitation is located at the origin, denoted as $1$. Annihilating this excitation by performing a qudit operation $X_{(2,{1}/{2})}^{-1}$ will create two new excitations: one located at $(1,0)$ and the other located at $(2,0)$, and these excitations are collectively denoted by $(a+1)x-ax^2$. This indicates that moving anyons by performing local operations may cause them to split, in contrast to the toric code, where they can be moved freely without creating additional excitations. This process is depicted in step 1 in Fig.~\ref{fig:quasi_fracton}.
  
This behavior resembles that of fractons, which split into multiple pieces when one attempts to move them. However, unlike the fractons, anyon excitations still permit one-to-one hopping over certain finite distances. We refer to this behavior as \textit{quasifractonic mobility}. As an explicit example, consider the DCY model with $N=4$ and $a=1$. In this case, the above process splits the original excitation into two anyons denoted as $3x^2+2x$. However, as depicted in Fig.~\ref{fig:quasi_fracton}, performing operations $X_{(3,{1}/{2})}^{2}$ and $X_{(4,{1}/{2})}$ successively annihilates all other excitations and creates a new excitation $x^4$, indicating that the original anyon is moved from $(0,0)$ to $(4,0)$.

\subsection{SET orders and mobility}
\label{sec5b}

A central question arises: Is it universally true for topological qudit BB codes that their topological excitations---though fractonlike under local moves---always permit one-to-one hopping over sufficiently large but finite distances? We answer this question in the affirmative below. Apparent discrepancies in the literature (e.g., the DCY model with non-coprime parameters $N=12$ and $a=2$~\cite{Delfino2023_2}) often stem from computational oversights in specific cases.

A general resolution of this mobility puzzle is naturally formulated in terms of SET order, which captures the interplay between symmetry and topological order. Here, the relevant symmetry is lattice translation, which is often essential for understanding fracton and fracton-like phenomena~\cite{Williamson2018, Pai2019, Dua2019}. The obstruction to elementary hopping [e.g., $(0,0)\to(1,0)$] arises because translations act nontrivially on anyons by permuting anyon types~\cite{Barkeshli2019, Chen2025}. 
Generically, moving an anyon $\eta\in\mathcal{C}$ by a lattice vector $(i,j)$ is forbidden if and only if $\eta$ and its translate $x^i y^j \eta$ represent distinct anyon types.

Thus, to identify all lattice displacements that admit one-to-one anyon transport, we introduce the mobility group $\Lambda_{\mathcal{C}}$. It is a subgroup of the translation group $\Lambda$, defined by
\begin{align}
    \Lambda_{\mathcal{C}}&=\{\lambda\in\Lambda|\lambda \eta=\eta,\forall \eta\in\mathcal{C}\}\notag\\
    &=\{\lambda\in\Lambda|\lambda-1\in(f,g)_R\},\label{eq:mobility_group}
\end{align}
where $(f,g)_R$ denotes the ideal of $R$ generated by polynomials $f$ and $g$.
Physically, $\Lambda_{\mathcal{C}}$ consists of precisely those translations that leave all anyon types invariant; consequently, it specifies exactly which lattice displacements admit one-to-one anyon transport implemented by string operators.
  
As a key aspect of $\Lambda_\mathcal{C}$, we define the anyon periodicities along the $\hat{x}$ and $\hat{y}$ directions as the minimal translation distances under which all anyon types remain invariant. These periodicities, denoted by $l_x$ and $l_y$, are given by
\begin{equation}
\begin{aligned}
l_x &=\min\{\alpha\in\mathbb{Z}_{>0} \mid x^{\alpha}-1\in (f,g)_R\},\\
l_y &=\min\{\beta\in\mathbb{Z}_{>0} \mid y^{\beta}-1\in (f,g)_R\}.
\end{aligned}
\label{eq:periodicity_general}
\end{equation}
Crucially, $l_x$ and $l_y$ always exist and are finite for any topological BB code, regardless of whether it is a qubit or qudit system. This follows from Corollary~\ref{cor:topo_finite}, which establishes that the set of anyon types $\mathcal{C}$ is finite. Consequently, $\mathcal{C}$ forms a finite $\Lambda$-set under lattice translations, ensuring that every translation orbit must close. This guarantees the existence of minimal positive integers $l_x$ and $l_y$ satisfying Eq.~\eqref{eq:periodicity_general}. This finiteness ensures that all excitations in topological BB codes---while appearing fracton-like at length scales below $l_x$ and $l_y$---possess long-range one-to-one mobility, fundamentally distinguishing them from true fractons.

For HGP codes, $l_x$ and $l_y$ completely characterize $\Lambda_{\mathcal{C}}$. In this case, the polynomials specifying $f$ and $g$ become univariate. As a result, the mobility group $\Lambda_C$ takes the form
\begin{equation}
    \Lambda_{\mathcal{C}}=\langle x^{l_x},y^{l_y}\rangle\cong l_x\mathbb{Z}\oplus l_y\mathbb{Z},\label{eq:mobility group with univariate polynomials}
\end{equation}
where $\langle x^{l_x},y^{l_y}\rangle$ denotes the subgroup of the translation group generated by $x^{l_x}$ and $y^{l_y}\in\Lambda$. For general BB codes with multivariate defining polynomials, $\langle x^{l_x},y^{l_y}\rangle$ may only be a subgroup of $\Lambda_{\mathcal{C}}$.

We conclude this subsection by examining the DCY model with parameters $N=12$ and $a=2$ as an explicit illustration of the general framework. This specific case is of particular interest because the existing literature \cite{Delfino2023_2} has questioned the finiteness of its periodicities, $l_x$ and $l_y$. For this model, the defining polynomials are $f=(x-1)(x-2)$ and $g=(y-1)(y-2)$, as given in Eq.~\eqref{eq:poly2}. To analyze the mobility, it suffices to consider the charge sector $\mathcal{Q}$ of the anyon module $\mathcal{C}=\mathcal{Q}\oplus \overline{\mathcal{Q}}$:
\begin{equation}
\mathcal{Q} = \frac{\mathbb{Z}_{12}[x^{\pm}, y^{\pm}]}{(f, g)} \cong \frac{\mathbb{Z}_{3}[x^{\pm}, y^{\pm}]}{(f, g)} \oplus \frac{\mathbb{Z}_{4}[x^{\pm}, y^{\pm}]}{(f, g)},
\end{equation}
where we utilize the decomposition $\mathbb{Z}_{12} \cong \mathbb{Z}_3 \oplus \mathbb{Z}_4$. We note 
\begin{enumerate}
    \item In the $\mathbb{Z}_3$ sector, $2 \equiv -1 \pmod 3$, so $f = (x-1)(x+1) = x^2-1$ and $g = y^2-1$.
    \item In the $\mathbb{Z}_4$ sector, $(x-2)$ is invertible since $(x-2)(x+2)x^{-2} \equiv 1 \pmod 4$. Similarly, $(y-2)$ is also invertible.
\end{enumerate}
Consequently, the charge module simplifies to:
\begin{equation}
\mathcal{Q} \cong \frac{\mathbb{Z}_{3}[x^{\pm}, y^{\pm}]}{(x^2 - 1, y^2 - 1)} \oplus \frac{\mathbb{Z}_{4}[x^{\pm}, y^{\pm}]}{(x - 1, y - 1)}.
\end{equation}
Under this direct sum decomposition of $\mathcal{Q}$, a unit $X$-type excitation at the origin is represented by $1\choose 1$. Translation along $\hat{x}$ acts via multiplication by $x$, yielding the orbit (in the quotient)
\begin{equation}
{1 \choose 1}\xrightarrow{\cdot x}{x \choose x}\equiv{x \choose 1}\xrightarrow{\cdot x}{x^{2} \choose x^{2}}\equiv{1 \choose 1}.
\end{equation}
This confirms  $l_x = 2$, contradicting the claim of non-existence of finite periods in Ref.~\cite{Delfino2023_2}. A quick check confirms this: The $Z$-string operator $Z_{(\frac{1}{2},0)} Z_{(-\frac{1}{2},0)}^3 Z_{(-\frac{3}{2},0)}^6$ successfully moves a unit $X$-type excitation from $(0,0)$ to $(2,0)$. Section~\ref{sec6b} extends this analysis to arbitrary DCY models, further confirming the general result on the existence of one-to-one anyon hopping processes.

\subsection{Ground state degeneracy}
\label{sec3c}
Another feature of translation symmetry enriched topological orders, absent in conventional topological orders, is the possible dependence of GSD on the system size~\cite{Watanabe2023, Chen2025}. In this section, we describe how to evaluate the GSD of the BB code on a torus of size $L_x\times L_y$ and discuss its relation to the underlying SET order in detail.

On a torus, the periodic boundary condition gives $x^{L_x}=y^{L_y}=1$. Thus, we work over the ring
\begin{equation}
    R_L \coloneqq\frac{R}{(x^{L_x}-1,y^{L_y}-1)}.
\end{equation}
Under this periodic boundary condition, the chain complex describing the $\mathbf{BB}_N(f,g)$ code becomes
\begin{equation}
R_{L}^2\xrightarrow{\partial_{L}=\left(\begin{array}{cc}
\partial_{AL} & 0\\
0 & \partial_{BL}
\end{array}\right)}R_{L}^{4}\xrightarrow{\varepsilon_L=\left(\begin{array}{cc}
0 & \varepsilon_{AL}\\
\varepsilon_{BL} & 0
\end{array}\right)}R_{L}^2,\label{eq:ccpbc}
\end{equation} 
where the maps are induced from, and take the same form as, those in Eqs.~\eqref{eq:pial} and \eqref{eq:eps}. Physically, $\partial_L$ and $\eps_L$ describe the stabilizer and excitation map of the BB code, respectively.

The GSD of the $\mathbf{BB}_{N}(f,g)$ code is given by $\gsd =\tr P$, where $P$ is the projector onto the ground-state subspace. Explicitly, $P$ can be expressed as
\begin{equation}
    P = \frac{1}{|S|}\sum_{s\in S}s
    =\frac{1}{|S_A|}\sum_{s_A\in S_A}s_A\cdot\frac{1}{|S_B|}\sum_{s_B\in S_B}s_B,\label{eq:PaPb}
\end{equation}
Here, $S = \langle S_A, S_B \rangle$ is the stabilizer group, where $S_A$ and $S_B$ are the groups generated by the operators $\{A_{(i,j)}\}$ and $\{B_{(i,j)}\}$ [defined in Eq.~\eqref{eq:Hamiltonian}] over all sites. The evaluation of this trace yields:
\begin{equation}
    \gsd(L_x,L_y) =\tr P =\left|\frac{R_{L}}{\left(f,g\right)}\right|^{2}.\label{eq:gsd1}
\end{equation}
In the following, we may use the shorthand $\gsd$, rather than $\gsd(L_x,L_y)$ for brevity. The detailed derivation, with special attention to ensure its validity for $\mathbb{Z}_{N}$ even when $N$ is a composite number, is presented in Appendix~\ref{app4}.

Clearly, the GSD depends on the system's size. According to Eq.~\eqref{eq:anyon}~and~\eqref{eq:gsd1}, it cannot exceed its total number of anyon types, and reaches this maximum value if and only if the system size is an integer multiple of the periodicity given in Eq.~\eqref{eq:periodicity_general}.

This reflects the SET nature of BB codes: The translation symmetry acts nontrivially on the anyons, so unless the system's size matches anyon periodicities, not all anyon types are preserved when they traverse non-contractible loops on a torus. Consequently, non-contractible loop operators do not exist for some anyon types, which in turn reduces the GSD.

Furthermore, the relation between $\gsd$ and the periodicities is quantitatively characterized by 
\begin{equation}
    \gsd(L_x,L_y)=\gsd\left(\gcd(L_x,l_x), \gcd(L_y,l_y)\right),
\end{equation}
where $\gcd$ denotes the greatest common divisor. The proof is analogous to that for the binary special case in Ref.~\cite{Chen2025}, requiring only the replacement of dimension arguments with cardinality arguments.
\section{Analytical results for explicit models\label{sec6}}
This section analytically characterizes the essential topological properties of the models introduced in Sec.~\ref{sec2}. First, in Sec.~\ref{sec6a}, we simply and concisely reproduce the known topological properties of the WCF model~\cite{Watanabe2023}. Next, in Sec.~\ref{sec6b}, our analysis of the DCY model resolves its mobility puzzle and provides what we believe to be an accurate determination of its GSD. Finally, in Secs.~\ref {sec6c}~and~Sec.~\ref{sec6d}, we apply Theorem~\ref{thm:anyon_type_relation} to determine the fusion rules for the product and BB codes with the toric layout. These results demonstrate the broad applicability of our framework for characterizing topological order.
\subsection{WCF model}
\label{sec6a}
We first consider the WCF model specified by the Laurent polynomials $f=x-a$ and $g=y-a$. We provide a simple and concise derivation of its topological properties originally obtained in~\cite{Watanabe2023}.
  
\subsubsection{Anyon types and fusion rules}

We specify the anyon type of the WCF model by invoking Theorem~\ref{thm:anyon_type_relation}. For each prime factor $p_i$ of the qudit dimension $N=\prod_{i=1}^np_i^{k_i}$, the corresponding $\mathbb{Z}_{p_i}$ WCF model is in the trivial phase if $p_i$ divides $a$, but has fusion rules $\mathbb{Z}_{p_{i}}^2$ if $a$ and $p_i$ are coprime. Thus, the fusion rules of the WCF model are specified by
\begin{equation}
    \mathcal{C}_{\operatorname{WCF}}\cong\mathbb{Z}_{N_a}^2\label{eq:fstc},
\end{equation}
where $N_a$ is the largest factor of the qudit dimension $N$ that is coprime to $a$. For example, if we take $N=12$ and $a=10$, we have $N_a=3$.

Consequently, the WCF model is in a trivial phase when $N_a=1$ and in a topologically ordered phase when $N_a>1$. Furthermore, the total number of anyon types is $|\mathcal{C}_{\operatorname{WCF}}|=N_a^2$.
  
\subsubsection{Mobility}

According to Eq.~\eqref{eq:periodicity_general}, the anyon types remain unchanged under translations by $\alpha$ and $\beta$ along $\hat{x}$ and $\hat{y}$ directions exactly when
\begin{subequations}
\begin{align}
        x^{\alpha}-1=0\bmod (N,f),\label{eq:l_x}\\
        y^{\beta}-1=0\bmod(N,g), \label{eq:l_y}
\end{align}    
\end{subequations}
where $\alpha,\beta\in\mathbb{Z}$. This gives the \emph{mobility conditions} of anyons. The anyon periodicities $l_x$ and $l_y$ are given by the smallest positive $\alpha$ and $\beta$ satisfying Eqs.~\eqref{eq:l_x}~and~\eqref{eq:l_y}, respectively.
  
We focus first on the $\hat{x}$ direction. We decompose the base ring $\mathbb{Z}_N=\mathbb{Z}_{\frac{N}{N_a}}\oplus\mathbb{Z}_{N_a}$ using the Chinese remainder theorem. It is clear that the following equalities hold:
\begin{subequations}
\begin{align}
        x^{\alpha}-1&=0 \bmod \left(\frac{N}{N_a},f\right),\\
        x^{\alpha}-1&=(a^\alpha-1)\bmod (N_a,f).
\end{align}
\end{subequations}
Thus, the mobility condition in Eq.~\eqref{eq:l_x} becomes $a^{\alpha}=1\bmod N_a$ since $x^{\alpha}-1$ needs to vanish on both direct summands. Euler's theorem guarantees the existence of such an $\alpha$. We denote the smallest positive integer $\alpha$ satisfying $a^{\alpha}=1\bmod N_a$ by $\ord_{N_a}(a)$. It is also the order of $a$ in the multiplicative group of units of $\mathbb{Z}_{N_a}$. Thus, the periodicity along the $\hat{x}$ direction is $\ord_{N_a}\left(a\right)$. 
  
The periodicity along the $\hat{y}$ direction can be obtained analogously. We conclude that the periodicities of the WCF model  along the  $\hat{x}$ and $\hat{y}$ directions are
\begin{equation}
    l_x=l_y=l_{\operatorname{WCF}}=\ord_{N_a}(a),
\end{equation}
and the mobility of anyons can be described by the subgroup $\Lambda_{{\operatorname{WCF}}}$ of $\Lambda$ as follows
\begin{align}
  \Lambda_{{\operatorname{WCF}}}=\langle x^{l_{\operatorname{WCF}}},y^{l_{\operatorname{WCF}}}\rangle\cong l_{\operatorname{WCF}}\mathbb{Z}\oplus l_{\operatorname{WCF}}\mathbb{Z}.\label{eq:pgb}
\end{align}
  
\subsubsection{Ground state degeneracy}

The GSD of the WCF model equals $|R_L/(f,g)|^2$ according to Eq.~\eqref{eq:gsd1}. By decomposing $R_L$ into a direct sum, we get
\begin{equation}
    \frac{R_L}{(f,g)}\cong\frac{{\mathbb{Z}_{\frac{N}{N_a}}[x^{\pm},y^{\pm}]}}{(x^{L_x}-1,y^{L_y}-1,f,g)}\oplus\frac{{\mathbb{Z}_{N_a}[x^{\pm},y^{\pm}]}}{(x^{L_x}-1,y^{L_y}-1,f,g)}.\label{eq:tc_decom}
\end{equation}
We will analyze the two direct summands separately. For the first direct summand, we note that all prime factors of $N/N_a$ also divide $a$ by definition. This implies that $a$ is nilpotent in the base ring $\mathbb{Z}_{{N}/{N_a}}$. However, the relations $f=x-a$ and $g=y-a$ require $a$ to be a unit in the quotient ring, as it is identified with the invertible elements $x$ and $y$. Thus, this summand is trivial.

For the second direct summand, since $a$ and $N_a$ are coprime, we can perform the long division. Dividing polynomials $x^{L_x}-1$ and $y^{L_y}-1$ by $f$ and $g$, respectively, reduces the second direct summand to
\begin{equation}
\begin{aligned}
        &\frac{\mathbb{Z}_{N_a}[x^{\pm},y^{\pm}]}{(x^{L_x}-1,y^{L_y}-1,f,g)}=\frac{\mathbb{Z}_{N_a}[x,y]}{(f,g,a^{L_x}-1,a^{L_y}-1)}\\
        &\cong\mathbb{Z}_{\gcd(N_a,a^{L_x}-1,a^{L_y}-1)}.\label{eq:tc_2}
\end{aligned}
\end{equation}
  
Combining these results, we obtain the GSD of the WCF model: 
\begin{equation}
    \gsd_{\operatorname{WCF}}=\left|\frac{R_{L}}{\left(f,g\right)}\right|^{2}=\left[\gcd(N_a,a^{L_x}-1,a^{L_y}-1)\right]^2.
\end{equation}
\subsection{DCY model}
\label{sec6b}

\begin{figure}[t]
    \centering
    \includegraphics[width=\linewidth]{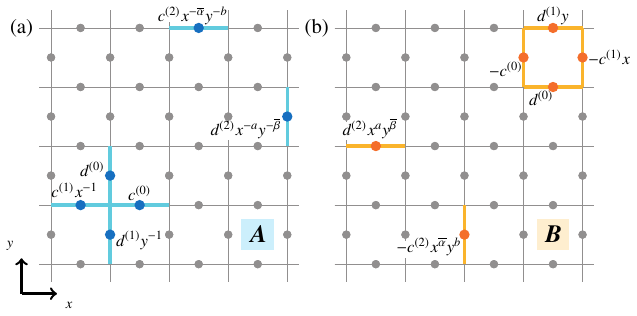}
    \caption{Illustration of the $X$-type (a) and the $Z$-type (b) stabilizers of the bivariate-bicycle code specified by $f=c^{(0)}+c^{(1)}x+c^{(2)}x^{\bar{\alpha}}y^{b}$ and $g=d^{(0)}+d^{(1)}y+d^{(2)}x^{a}y^{\bar{\beta}}$. On the lattice, a unit cell at the vertex $(i,j)$ consists of two qudits on the edges $(i,j+\frac{1}{2})$ and $(i+\frac{1}{2},j)$. The deep blue and orange circles on the lattice edges represent the constituent $X$ and $Z$ Pauli operators, respectively. Each operator is labeled by a monomial in which the exponents of $x$ and $y$ denote the unit cell coordinates, and the coefficient gives the operator's power.}
    \label{fig:bbcode}
\end{figure}

Now we turn to the DCY model specified by polynomials $f=(x-a)(x-1)$ and $g=(y-1)(y-a)$. We begin by determining its anyon types and their fusion rules. We then resolve its mobility puzzle by directly evaluating its mobility group. Finally, we present what we believe is the first accurate evaluation of the GSD of the DCY model using the polynomial framework.
  
\subsubsection{Anyon types and fusion rules}---In analogy with the WCF model, applying Theorem~\ref{thm:anyon_type_relation}, the fusion rules of the DCY model are given by
\begin{equation}
    \mathcal{C}_{\operatorname{DCY}}=\mathbb{Z}_{N_a}^6\oplus\mathbb{Z}_N^2.
\end{equation}
The DCY model is always in the topologically ordered phase, and the total number of anyon types is $|\mathcal{C}_{\operatorname{DCY}}|=N_a^6N^2$.
  
\subsubsection{Mobility}

Similar to the WCF model, the periodicities $l_x$ and $l_y$ are the smallest positive $\alpha$ and $\beta$ satisfying the mobility conditions in Eqs.~\eqref{eq:l_x}~and~\eqref{eq:l_y}. We begin by concentrating on the $\hat{x}$ direction. It is evident that
\begin{subequations}
\begin{align}
    x^{\alpha}-1&=0\bmod \left(\frac{N}{N_a},f\right),\label{period_BQa}\\
    \quad x^{\alpha}-1&= G_{\alpha}(x-1) \bmod (N_a,f).\label{period_BQb}
\end{align}
\end{subequations}
where 
\begin{equation}
G_{\alpha}\coloneqq\sum_{m=0}^{\alpha-1}a^m  
\end{equation}
denotes the finite geometric series in the variable $a$, up to degree $\alpha-1$. We decompose the base ring into $\mathbb{Z}_N=\mathbb{Z}_{\frac{N}{N_a}}\oplus\mathbb{Z}_{N_a}$. According to Eqs.~\eqref{period_BQa}~and~\eqref{period_BQb}, the mobility condition in Eq.~\eqref{eq:l_x} is trivial on the first direct summand, and reduces to $G_{\alpha}=0\bmod N_a$ on the second. This condition implies $a^{\alpha}=1\bmod N_a$. Therefore, the periodicity $l_x$ must be an integer multiple of $\ord_{N_a}(a)$. Additionally, using the property $a^{m+\ord_{N_a}(a)}=a^m\bmod N_a$, it can be shown that $nG_{\ord_{N_a}(a)}=G_{n\ord_{N_a}(a)}\bmod N_a$ for any positive integer $n$. Thus, basic algebraic manipulation gives
\begin{equation}
    l_{\operatorname{DCY}}=\ord_{N_a}\left(a\right)\frac{N_a}{\gcd(G_{\ord_{N_a}(a)},N_a)}.
\end{equation}
  
Therefore, the periodicities of the DCY model along $\hat{x}$ and $\hat{y}$ directions are 
\begin{equation}
    l_x=l_y=l_{\operatorname{DCY}}.
\end{equation}
and the mobility group $\Lambda_{{\operatorname{DCY}}}$ whose elements preserve all anyon types is
\begin{equation}
    \Lambda_{{\operatorname{DCY}}}=\langle x^{l_{\operatorname{DCY}}},y^{l_{\operatorname{DCY}}}\rangle\cong l_{\operatorname{DCY}}\mathbb{Z}\oplus l_{\operatorname{DCY}}\mathbb{Z}.
\end{equation}

The mobility group $\Lambda_{{\operatorname{DCY}}}$ gives a precise characterization of the DCY model's anyon mobility and resolves the anyon mobility puzzle of the DCY model. It never becomes trivial, so the one-to-one hoppings are always permitted, aligning with the argument in Sec.~\ref{sec5b}.

\subsubsection{Ground state degeneracy}

The GSD of the DCY model is given by $|R_L/(f,g)|^2$ with $f,g$ satisfying Eq.~\eqref{eq:poly2}. In this case, the quotient module $R_L/(f,g)$ still decomposes according to the direct sum in Eq.~\eqref{eq:tc_decom}, so we analyze the two summands separately.

For the first summand, as established previously, $a$ is nilpotent in the base ring. Thus, $x-a$ and $y-a$ are units in the first direct summand. It follows that
\begin{equation}
    \frac{\mathbb{Z}_{\frac{N}{N_a}}[x^{\pm},y^{\pm}]}{(x^{L_x}-1,y^{L_y}-1,f,g)}\cong \frac{\mathbb{Z}_{\frac{N}{N_a}}[x^{\pm},y^{\pm}]}{(x-1,y-1)}\cong\mathbb{Z}_{\frac{N}{N_a}}.\label{eq:quad_gsd_first_direct_sum}
\end{equation}
For the second summand, $N_a$ is coprime to $a$, so this allows us to reduce the second direct summand by the long division:
\begin{equation}
    \frac{\mathbb{Z}_{N_a}[x^{\pm},y^{\pm}]}{(x^{L_x}-1,y^{L_y}-1,f,g)}=\frac{\mathbb{Z}_{N_a}[x,y]}{(f,g,h_x,h_y)}\cong \frac{\mathbb{Z}_{N_a}[x]}{(f,h_x)}\otimes_{\mathbb{Z}_{N_a}}\frac{\mathbb{Z}_{N_a}[y]}{(g,h_y)},\label{eq:quad_gsd_second_direct_sum}
\end{equation}
where we define $h_x=G_{L_x}(x-1)$ and $h_y=G_{L_y}(y-1)$ for convenience. The analysis of each tensor product factor yields the following isomorphisms:
\begin{equation}
\begin{aligned}
    \frac{\mathbb{Z}_{N_a}[x]}{(f,h_x)}&\cong\mathbb{Z}_{N_a}\oplus\mathbb{Z}_{\gcd(N_a,G_{L_x})},\\
    \frac{\mathbb{Z}_{N_a}[y]}{(g,h_{y})}&\cong\mathbb{Z}_{N_a}\oplus\mathbb{Z}_{\gcd(N_a,G_{L_y})},\label{quad_gsd_univariate}
\end{aligned}
\end{equation}
Substituting Eq.~\eqref{quad_gsd_univariate} into Eq.~\eqref{eq:quad_gsd_second_direct_sum} reduces the second direct summand of Eq.~\eqref{eq:tc_decom} to
\begin{align}
    &\frac{\mathbb{Z}_{N_a}[x^{\pm},y^{\pm}]}{(x^{L_x}-1,y^{L_y}-1,f,g)}\cong \frac{\mathbb{Z}_{N_a}[x]}{(f,h_x)}\otimes_{\mathbb{Z}_{N_a}}\frac{\mathbb{Z}_{N_a}[y]}{(g,h_y)}\notag\\
    &\cong\mathbb{Z}_{N_a}\oplus\mathbb{Z}_{\gcd(N_a,G_{L_x})}\oplus\mathbb{Z}_{\gcd(N_a,G_{L_y})}\oplus\mathbb{Z}_{\gcd(N_a,G_{L_x},G_{L_y})}.
\end{align}
  
Consequently, the GSD of the DCY model is
\begin{equation}
    \gsd_{\operatorname{DCY}}=\left[N\gcd(N_a,G_{L_x})\gcd(N_a,G_{L_y})\gcd(N_a,G_{L_x},G_{L_y})\right]^2. \label{eq:DCY_GSD}
\end{equation}
We believe that this is an accurate calculation of the GSD of the DCY model. It corrects the previous result reported in Ref.~\cite{Delfino2023_2}.

As an explicit check, consider the case $N=2$ and $a=1$ (so $f=x^2-1$ and $g=y^2-1$) on a torus with $L_x=L_y=l$. In this case, the model reduces to four decoupled toric codes supported on the four sublattices distinguished by the parities of $(i,j)$: (i) $i$ odd, $j$ odd; (ii) $i$ even, $j$ odd; (iii) $i$ odd, $j$ even; and (iv) $i$ even, $j$ even. When $l$ is even, periodic boundary conditions preserve these sublattices, yielding $\operatorname{GSD}=4^4=256$; when $l$ is odd, periodic boundary conditions identify them, leaving a single toric code with $\operatorname{GSD}=4$. Our Eq.~\eqref{eq:DCY_GSD} correctly captures this system-size dependence. In contrast, Ref.~\cite{Delfino2023_2} [Eq.~(27) therein] predicts $\operatorname{GSD}=256$ for this example independent of $l$, and thus misses the odd-$l$ behavior.

\subsection{Hypergraph product codes\label{sec6c}}
We obtain the anyon types and their fusion rules for the $\mathbb{Z}_N$ qudit HGP code specified by univariate polynomials $f(x)$ and $g(y)$ in the following. According to Theorem~\ref{lem:topo_prime_factors}, the HGP code is not topological if and only if $fg=0$ and $f+g$ is not a monomial over $\mathbb{Z}_p[x^\pm,y^\pm]$ for all primes $p\mid N$.
For prime qudit dimensions, the fusion rules of HGP codes can be easily obtained by long division. Consequently, the fusion rules of HGP codes with an arbitrary qudit dimension $N$ can be directly obtained using Theorem~\ref{thm:anyon_type_relation}. The result is given by
\begin{equation}
    \mathcal{C}_{\operatorname{Pord}}\cong\bigoplus_{i=1}^n\mathbb{Z}_{p_i^{k_i}}^{2\operatorname{deg}_x(f_{p_i})\operatorname{deg}_y(g_{p_i})}.\label{eq:anyon_type_uni}
\end{equation}
Here, $f_{p_i}=f\bmod p_i$ and $g_{p_i}=g\bmod p_i$ denote the reductions of Laurent polynomials $f$ and $g$ in $\mathbb{Z}_{p_i}[x^{\pm},y^{\pm}]$. The degrees of Laurent polynomials, denoted as $\operatorname{deg}_x$ and $\operatorname{deg}_y$, are the differences between the highest and lowest exponents of the corresponding variable. For example, the $x$ degree of the polynomial $h=\bar{x}+1+x$ is 2. This result is consistent with the fusion rules for the WCF and DCY models derived in the preceding two subsections.
\subsection{BB codes with toric layouts\label{sec6d}}

\begin{figure}[t]
    \centering
    \includegraphics[width=\linewidth]{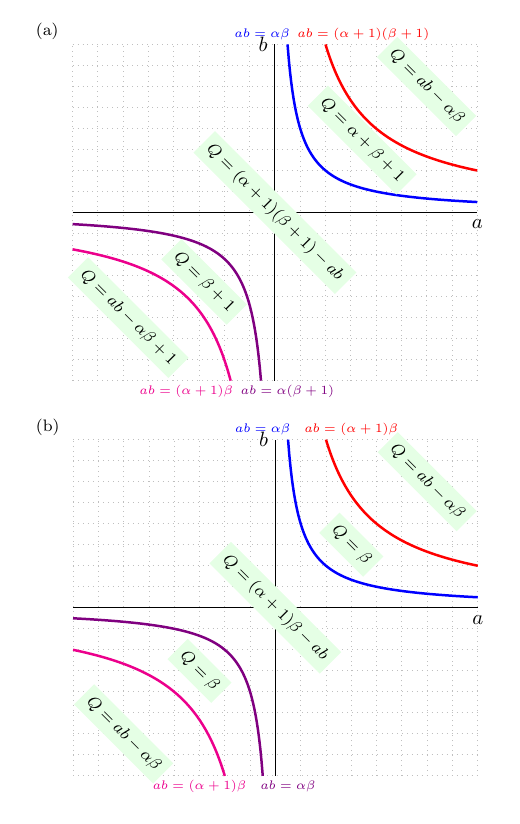}
    \caption{The topological index of the $\mathbf{BB}_p(f,g)$ code with $f$ and $g$ from Eq.~\eqref{eq:tc_layout} for two special cases: (a) all coefficients are non-zero modulo $p$. (b) only $d^{(1)}$ modulo $p$ is zero. For both cases, the parameter space is divided into five regimes by four hyperbolas, and the corresponding expression of the topological index in each regime is indicated. Parameters are restricted to $\beta\ge\alpha\ge0$ for (a), and to $\alpha,\beta\ge 0$ for (b), as all other values can be reduced to this range by suitable origin shifts and axis reflections.}
    \label{fig:topological_index}
\end{figure}

Here we show how the application of BKK theorem~\cite{Chen2025} can be extended—via Theorem~\ref{thm:anyon_type_relation}—to determine the anyon types and their fusion rules in $\mathbb{Z}_N$ BB codes. We focus on the $\mathbb{Z}_N$ BB codes with toric layouts, i.e., codes whose check operators largely mirror those of the toric code, except for the involvement of two additional, spatially separated qudits, as illustrated in Fig.~\ref{fig:bbcode}. These $\mathbf{BB}_N(f,g)$ codes are specified by polynomials of the following form
\begin{equation}
    f=c^{(0)}+c^{(1)}x+c^{(2)}x^{\bar{\alpha}}y^{b},\quad g=d^{(0)}+d^{(1)}y+d^{(2)}x^{a}y^{\bar{\beta}},\label{eq:tc_layout}
\end{equation}
where all coefficients are in $\mathbb{Z}_N$, and the parameters are denoted as $\bar{\alpha}=-\alpha$ and $\bar{\beta}=-\beta$ for convenience. In this case, the fusion rules are described by 
\begin{equation}
    \mathcal{C}_N\cong\bigoplus_{i=1}^n \mathbb{Z}_{p_i^{k_i}}^{2Q_i},\label{eq:tc_layout_0.5}
\end{equation}
where $Q_i$ denotes the topological index of the $\mathbf{BB}_{p_i}(f,g)$ code for each $p_i\mid N$. These topological indices can be obtained systematically from the Bernstein-Khovanskii-Kushnirenko theorem.
  
To illustrate the evaluation of the topological index, we use the $\mathbf{BB}_{p^k}(f,g)$ code as a concrete example, with $f$ and $g$ defined in Eq.~\eqref{eq:tc_layout}. In this case, the only prime factor is $p$, so the fusion rules for the anyon types are
\begin{equation}
    \mathcal{C}_{p^k}\cong \mathbb{Z}_{p^k}^{2Q},
\end{equation}
where $Q$ is the topological index of the $\mathbf{BB}_{p}(f,g)$ qudit BB code. In particular, Fig.~\ref{fig:topological_index}(a) depicts the topological index $Q$ when
\begin{equation}
    f_p=c^{(0)}_p+c^{(1)}_px+c^{(2)}_px^{\bar{\alpha}}y^{b},\quad g_p=d^{(0)}_p+d^{(1)}_py+d^{(2)}_px^{a}y^{\bar{\beta}}.\label{eq:tc_layout_case_1}
\end{equation}
Here, all coefficients in Eq.~\eqref{eq:tc_layout_case_1} are nonzero, and the subscript $p$ in each coefficient denotes modulo $p$. Likewise, Fig.~\ref{fig:topological_index}(b) shows the topological index $Q$ when
\begin{equation}
f_p=c^{(0)}_p+c^{(1)}_px+c^{(2)}_px^{\bar{\alpha}}y^{b},\quad g_p=d^{(0)}_p+d^{(2)}_px^{a}y^{\bar{\beta}}.\label{eq:tc_layout_case_2}
\end{equation}
Similarly, all coefficients in Eq.~\eqref{eq:tc_layout_case_2} are nonzero. For the general composite qudit dimension 
$N$, the anyon types are determined by applying this procedure to each of its prime factors, as prescribed by the direct sum in Eq.~\eqref{eq:tc_layout_0.5}.
\section{Computational algebraic methods\label{sec7}}
While analytically deriving the topological order and its symmetry-enriched properties of BB codes is challenging, these properties can be systematically obtained by computational algebraic methods using Gr\"{o}bner basis techniques of the polynomial rings defined over $\mathbb{Z}$~\cite{Moller1988, Kandri-Rody1988, Becker1993, Adams1994}. 

\begin{table*}[ht]
            \centering
            \begin{tabular}{c|c|c|c|c|c|c}
            \hline
            \hline
            $N$ & $f$ & $g$ & Fusion rules & Periodicity & Mobility group &GSD\\
            \hline
            4 & $1+x+y^2$ & $1+y+x^3$ & $\mathbb{Z}_4^8$ & $l_x=10,l_y=30$& $\langle x^{10}, x^{6}y^6\rangle,\langle y^{30}, x^{2}y^{12}\rangle$ & $1$\\
            \hline
            4 & $2x^3+y^2+y$ & $y^3+x^2+x$ &$\mathbb{Z}_4^4$&$l_x=3,l_y=1$& $\langle x^3,y\rangle$ &$4^4$\\
            \hline
            4 & $x^5y^{-1}+y+1$ & $y^5x^{-1}+x+1$ &$\mathbb{Z}_4^{48}$&$l_x=l_y=19530$& $\langle x^{19530},x^{12090}y^{930}\rangle,\langle y^{19530},x^{930}y^{12090}\rangle$&$1$\\
            \hline
            4 & $2+x+xy^2$ & $1+2y+x^4y^3$ &$\mathbb{Z}_4^{16}$&$l_x=16,l_y=4$& $\langle x^{16},x^8y^2\rangle,\langle y^4,x^8y^2\rangle$ &$4^3$\\
            \hline
            4 & $1+x+x^{-1}y^3$ & $1+y+2x^3y^{-2}$ &$\mathbb{Z}_4^4$&$l_x=3,l_y=1$& $\langle x^3,y\rangle$ &$4^4$\\
            \hline
            6 & $2x^3+y^2+y$ & $2y^3+x^2+x$ &$\mathbb{Z}_2^2\oplus\mathbb{Z}_3^{16}$&$l_x=l_y=104$& $\langle x^{104},x^{91}y^{13}\rangle,\langle y^{104},x^{13}y^{91}\rangle$ &$4$\\
            \hline
            8 & $x^3+y^2+y$ & $y^3+x^2+x$ &$\mathbb{Z}_8^{16}$&$l_x=l_y=48$&$\langle x^{48},x^{24}y^{24}\rangle,\langle y^{48},x^{24}y^{24}\rangle$&$8^9\times 2$\\
            \hline
            9 & $4+6x+x^{-1}y^{2}$ & $7+5y+2x^{-3}y^{-4}$& $\mathbb{Z}_9^{22}$& $l_x=39348,l_y=78696$ & $\langle x^{39348},x^{39342}y^{12}\rangle,\langle y^{78696},x^{6}y^{78684}\rangle$ & $9^4$\\
            \hline
            12 & $4x^3+y^2+y$ & $6y^5+x^2+x$ &$\mathbb{Z}_4^2\oplus\mathbb{Z}_3^{4}$&$l_x=2,l_y=8$& $\langle x^2,xy^4\rangle, \langle y^8,xy^4\rangle$ &16\\
            \hline
            12 & $3+4x+x^{-1}y^{-2}$ & $6+y+2x^2y$ &$\mathbb{Z}_3^{8}$&$l_x=2,l_y=4$& $\langle x^2,y^4\rangle$ & 1\\      
            \hline
            \hline
            \end{tabular}
            \caption{Topological properties of various $\mathbb{Z}_N$ bivariate-bicycle (BB) codes, including their defining polynomials $f$ and $g$, fusion rules, periodicities ($l_x, l_y$), mobility groups, and ground state degeneracies (GSD) computed for system size $L_x = L_y = 18$. Here, $\langle\bullet,\bullet\rangle$ denotes the subgroup of the translation group $\Lambda$ generated by a pair of elements in $\Lambda$.}
            \label{tab:topological properties of BB code}
\end{table*}

\subsection{Gr\"{o}bner bases over the ring of integers}
\subsubsection{Definition of Gr\"{o}bner bases}
In this section, we give a brief introduction to Gr\"{o}bner bases over $\mathbb{Z}$. For multivariate polynomials, the concept of degree for univariate polynomials is generalized to the monomial order~\footnote{Formally, the monomial order is a total well order $>$ of monomials satisfying that $m>1$ if $m\neq 1$, and that $m_1n>m_2n$ if $m_1>m_2$}. Once a monomial order is chosen, the leading term $\operatorname{Lt}(h)$ of a non-zero polynomial $h$  is the term whose monomial is largest with respect to this order. Over $\mathbb{Z}$, it is crucial to emphasize the distinction between the leading term, which includes the leading coefficient, and the leading monomial of a polynomial. For instance, given $h=2x^2+x\in\mathbb{Z}[x]$, its leading term is $\operatorname{Lt}(h)=2x^2$, whereas its leading monomial is $\operatorname{Lm}(h)=x^2$ with the leading coefficient $\operatorname{Lc}(h)=2$.

Given a monomial order, a Gr\"{o}bner basis $G$ of an ideal $I$ is a finite subset of $I$ such that
\begin{equation}
    \operatorname{Lt}(G)=\operatorname{Lt}(I),\label{eq:definition of Grobner basis}
\end{equation}
where $\operatorname{Lt}(G)$ and $\operatorname{Lt}(I)$ denote the ideals generated by the sets of leading terms of all elements in $G$ and $I$, respectively. 

To apply the Gr\"{o}bner basis techniques to the Laurent polynomial ring $R=\mathbb{Z}_N[x^{\pm},y^{\pm}]$, it is convenient to map the Laurent polynomial ring to a quotient of a polynomial ring using the following isomorphism
\begin{equation}
    \mathbb{Z}_N[x^{\pm},y^{\pm}]\cong\frac{\mathbb{Z}[x,y,t,u]}{(N,xt-1,yu-1)}.\label{eq:isomorphism_of_Laurent_ring}
\end{equation}
\subsubsection{Crucial properties of Gr\"{o}bner bases\label{sec7a2}}
In this section, we discuss some crucial properties of Gr\"{o}bner bases that are essential for determining the topological condition, fusion rules, periodicities, and the GSD of BB codes.

\emph{Ideal membership problems}---Like the field case, given a Gr\"{o}bner basis $G$ for any ideal $I$, the polynomial $h$ lies in $I$ if and only if it can be reduced to zero by $G$ through the division algorithm~\cite{Moller1988}, i.e.,
\begin{equation}
    h=0\bmod G.\label{eq:condition_of_ideal_membership}
\end{equation}

\emph{Normal forms and quotient rings}---Unlike over a field, even with a fixed monomial order and corresponding Gr\"{o}bner basis, the higher-order terms may not be fully reduced by the division algorithm over Noetherian rings, since the leading coefficients of the Gr\"{o}bner basis elements need not be invertible. Nevertheless, for the ring of integers, one can still define a set of unique normal forms (remainders) to represent the quotient by the ideal $I$~\cite{Zacharias1978}.
  
In the following, we present a practical and intuitive construction of the normal forms. For their existence and uniqueness, we refer the reader to~\cite{Adams1994}. For our purpose, we focus on the polynomial ring $\mathbb{Z}[x,y,t,u]$ in the following. Let $G$ be a Gr\"{o}bner basis for the ideal $I$. Then, the set of normal forms is given by
\begin{equation}
    \{r|r=\sum_{\mathbf{w}}c_{\mathbf{w}}\mathcal{X}^{\mathbf{w}},c_{\mathbf{w}}\in\mathbb{Z}_{N_{\mathbf{w}}}\},\label{eq:effective_coset_representative}
\end{equation} 
where $\mathcal{X}^{\mathbf{w}}=x^{\omega_x}y^{\omega_y}t^{\omega_t}u^{\omega_u}$ denotes a monomial in the polynomial ring $\mathbb{Z}[x,y,t,u]$ with $\mathbf{w}=(\omega_x,\omega_y,\omega_t,\omega_u)\in\mathbb{Z}_{\geq0}^4$. The range of the coefficient $c_{\mathbf{w}}$ is $\mathbb{Z}_{N_{\mathbf{w}}}$, with $N_{\mathbf{w}}$ defined as
\begin{equation}
    N_{\mathbf{w}}=\gcd(C_{\mathbf{w}}),\quad C_{\mathbf{w}}:=\{\operatorname{Lc}(h)|\mathcal{X}^{\mathbf{w}}=0\bmod\operatorname{Lm}(h),h\in G\},\label{eq:definition of G_M}
\end{equation}
and $\gcd(C_{\mathbf{w}})$ denotes the greatest common divisor of all elements in $C_{\mathbf{w}}$~\footnote{If $C_{\mathbf{w}}=\varnothing$, then we take $N_{\mathbf{w}}=N$.}. This is because the leading term $c\mathcal{X}^{\mathbf{w}}$ of a polynomial can be reduced by $G$ through the division algorithm if and only if its coefficient lies in the ideal $(C_{\mathbf{w}})_{\mathbb{Z}}$. 

It's crucial to note that although these normal forms of the coset are unique, the coefficients $c_{\mathbf{w}}$ in Eq.~\eqref{eq:effective_coset_representative} do not necessarily form independent cyclic components of the quotient ring's additive group. In other words, the additive group structure is not simply a direct sum of $\mathbb{Z}_{N_{\mathbf{w}}}$. Thus, unlike the field case, this set of normal forms does not immediately unveil the additive group structure of the quotient ring with respect to $I$.
\subsection{Topological properties of BB codes from Gr\"{o}bner basis}
In this section, we examine how to determine the essential topological properties of BB codes using the Gr\"{o}bner basis techniques.

As we discussed in Sec.~\ref{sec4a}, the $\mathbf{BB}_N(f,g)$ code exhibits topological order if and only if $\mathcal{Q}=R/(f,g)$ is finite. Eq.~\eqref{eq:isomorphism_of_Laurent_ring} enables us to consider the ideal $I=(N,f,g,xt-1,yu-1)$ over $\mathbb{Z}[x,y,t,u]$. Then, $\mathcal{Q}$ is identified with $\mathbb{Z}[x,y,t,u]/I$. Thus, after computing the Gr\"{o}bner basis $G$ of $I$, the topological condition of the $\mathbf{BB}_N(f,g)$ code is directly given by the finiteness of the set of normal forms in Eq.~\eqref{eq:effective_coset_representative}.
  
The anyon types and their fusion rules are specified by $\mathcal{C}=\mathcal{Q}\oplus\overline{\mathcal{Q}}$ in Sec.~\ref{sec4b}. Thus, we still identify $\mathcal{Q}$ with $\mathbb{Z}[x,y,t,u]/I$ as above. As discussed in Sec.~\ref{sec7a2}, the normal forms of $I$ with respect to a Gr\"{o}bner basis $G$ do not directly reveal the additive group of $\mathcal{Q}$, so the determination of the additive group of $\mathcal{Q}$ for the $\mathbf{BB}_N(f,g)$ code by Gr\"{o}bner basis is less straightforward than for binary BB codes, where the unique remainders form a basis of the vector space. Our strategy here is to invoke Theorem~\ref{thm:anyon_type_relation}: The total number of anyon types $|\mathcal{Q}|$ is directly obtained by counting the cardinality of the set of normal forms in Eq.~\eqref{eq:effective_coset_representative}. Then, by applying Theorem~\ref{thm:anyon_type_relation}, the fusion rules are directly given by comparing the prime factorizations of $|\mathcal{Q}|$ and the qudit dimension $N$. Specifically, according to Theorem~\ref{thm:anyon_type_relation}, if $N=\prod_{i=1}^{n}p_i^{k_i}$, the prime factorizations of $|\mathcal{Q}|$ is given by $|\mathcal{Q}|=\prod_{i=1}^{n}p_i^{j_i}$ with $k_i\mid j_i$ for some $j_i$, and the fusion rules are specified by
\begin{equation}
    \mathcal{C}\cong\bigoplus_{i=1}^n\mathbb{Z}_{p_i^{k_i}}^{2j_i/k_i}.
\end{equation}

Moreover, the anyon periodicity and the mobility group of the $\mathbf{BB}_N(f,g)$ code are given by Eqs.~\eqref{eq:mobility_group}~and~\eqref{eq:periodicity_general}, respectively. Determining them requires testing whether $x^{\alpha}y^{\beta}-1$ with $\alpha,\beta\in \mathbb{Z}$ lies in $(f,g)$ in $R$. This is an ideal membership problem. Using the isomorphism in Eq.~\eqref{eq:isomorphism_of_Laurent_ring}, it is equivalent to the ideal membership problem of $I=(N,f,g,xt-1,yu-1)$ in $\mathbb{Z}[x,y,t,u]$. Thus, the anyon periodicity and the mobility group are directly obtained by testing if $x^{\alpha}y^{\beta}-1$ reduces to zero in $I$ by the division algorithm with respect to a Gr\"{o}bner basis $G$.

Finally, the ground state degeneracy of the BB code is given by the size of $R_L/(f,g)$. This value is obtained by counting the number of normal forms with respect to the ideal $ I' = (N,f,g,x^{L_x}-1,y^{L_y}-1,xt-1,yu-1)$ in the polynomial ring. This can be computed via a Gr\"{o}bner basis, analogously to the analysis of the topological condition and fusion rules.

\subsection{Examples}
The Gr\"{o}bner basis of polynomial rings defined over $\mathbb{Z}$ can be computed by software like Macaulay 2~\cite{M2} and Singular~\cite{DGPS}. As an explicit example, we consider a $\mathbb{Z}_{12}$ qudits BB code with the toric layout defined by
\begin{equation}
f=1+x+6\bar{x}y^5,\quad g=1+y+4x^3\bar{y},
\end{equation}
  
To determine the topological condition and fusion rules of this code, by applying Eq.~\eqref{eq:isomorphism_of_Laurent_ring}, we identify $\mathcal{Q}$ with
\begin{equation}
    \mathcal{Q}\cong\frac{\mathbb{Z}[x,y,t,u]}{(12,xt-1,yu-1,1+x+6ty^5,1+y+4x^3u)}.
\end{equation}
We choose the lexicographic order with $u>t>y>x$ here. Macaulay 2 computation gives that the Gr\"{o}bner basis with respect to the ideal is
\begin{align}
    G=\{12, x-5,3y+3,y^2+y-4,t-5,u-y-4\}.
\end{align}
The leading terms of the elements in the Gr\"obner basis $G$ are $\operatorname{Lt}(G)=\{12,x,3y,y^2,t,u\}$. The corresponding leading monomials are $\operatorname{Lm}(G)=\{1,x,y,y^2,t,u\}$. This implies that any monomial divisible by $x$, $y^2$, $t$, or $u$ can be reduced by $G$. The only monomials that are not divisible by any of these are $1$ and $y$. These form the basis for our set of normal forms. 

We now determine the allowed integer coefficients for these basis monomials according to the rules in Eq. \eqref{eq:definition of G_M}: For the monomial $1$, the only relevant basis element is $12$, and its coefficient $c_0$ lies in $\mathbb{Z}_{12}$; For the monomial $y$, the relevant basis element is $3y+3$ and $12$, and its coefficient $c_y$ lies in $\mathbb{Z}_{\gcd(3,12)}=\mathbb{Z}_3$. Therefore, the set of normal forms is given by
\begin{equation}
    \{c_0+c_yy|c_0\in\mathbb{Z}_{12},c_{y}\in\mathbb{Z}_3\}.\label{eq:example}
\end{equation}
  
Since the set of normal forms is finite, the model exhibits topological order. Additionally, counting the cardinality of the set of normal forms in Eq.~\eqref{eq:example} yields that $|\mathcal{Q}|=36$. Then, by comparing the prime factorizations of $N$ and $|\mathcal{Q}|$,  the fusion rules of anyon types of the model are given by $\mathcal{C}\cong\mathbb{Z}_4^{2}\oplus \mathbb{Z}_3^{4}$.

The ideal membership tests in Macaulay 2 show that the periodicities along the $\hat{x}$ and $\hat{y}$ directions are $l_x=2$ and $l_y=8$, and that the mobility group is $\langle x^2,xy^4\rangle$, or equivalently $\langle y^8,xy^4\rangle$. Finally, applying the same Gr\"{o}bner basis methods to the ideal for a finite torus, we find that for the code of size $L_x=L_y=18$, the ground state degeneracy is 16.

The simple sample code of this example is provided in Appendix~\ref{app5}, which also demonstrates how the topological condition can be verified directly by definition. The topological properties of several other BB codes determined using this method are summarized in Table~\ref{tab:topological properties of BB code}.
\section{Conclusion\label{sec8}}
In this work, we explored several generalizations of the toric code model for qudits. By utilizing the polynomial representation, we systematically characterized essential topological properties, including the topological condition, the anyon types and their fusion rules of qudit BB codes. We established a convenient finiteness criterion for the topological condition of qudit BB codes. Our analysis revealed that both the topological condition and the fusion rules for general qudit dimensions can be understood by examining prime-dimensional BB codes. This insight allowed us to employ an algebraic-geometric framework, specifically leveraging the BKK theorem, to systematically evaluate the anyon types and their fusion rules. Additionally, we examined that unconventional phenomena in qudit BB codes, including the quasifractonic behavior and size-dependent ground state degeneracy, originate from their symmetry-enriched topological order. Furthermore, we analytically derived the key topological properties of the WCF and DCY models. Our work resolves the mobility puzzle and presents an accurate analytical computation of the ground state degeneracy for the DCY model. Lastly, for more complicated BB codes, we proposed a computational algebraic method to calculate their topological properties based on the Gr\"{o}bner bases over the ring of integers. We hope that our framework will facilitate future work on broader classes of qudit stabilizer codes and their topological and symmetry-enriched properties.

\begin{acknowledgments}
This work is supported by the National Natural Science Foundation of China (Grants No.~12522502, No.~12474145, No.~12275331, and No.~12447101), the Strategic Priority Research Program of Chinese Academy of Sciences (Grant No. XDB1680000), the Fundamental and Interdisciplinary Frontier Research Priority Program of Chinese Academy of Sciences (Grant No.~XDB0920000), and the Innovation Program for Quantum Science and Technology (Grant No.~2021ZD0301602).
\end{acknowledgments}

\appendix

\section{Proofs of Theorem~\ref{lem:topo_prime_factors} and Corollary~\ref{cor:topo_finite}}
\label{app1}
This appendix provides the algebraic background and formal proofs for Theorem~\ref{lem:topo_prime_factors} and Corollary~\ref{cor:topo_finite}.

\subsection{Preliminaries}
We begin by establishing notation and recalling key theorems. Let $N = \prod_{i}p_i^{k_i}$ be the prime factorization of $N$. The ring of Laurent polynomials over $\mathbb{Z}_N$ is denoted $R=\mathbb{Z}_N[x^{\pm},y^{\pm}]$. For any prime factor $p$ of $N$, let $R_p \coloneqq R/pR \cong \mathbb{Z}_p[x^{\pm},y^{\pm}]$, and we denote the image of an element $r \in R$ in $R_p$ by $r_p$.

A crucial tool is McCoy's theorem, which characterizes zero divisors in polynomial rings.
\begin{thm}[McCoy's Theorem, \cite{Mccoy1942, Mccoy1957}]
Let $S$ be a commutative ring. A polynomial $h \in S[x_1, \dots, x_m]$ is a zero divisor if and only if there exists a non-zero element $s \in S$ such that $s\cdot h=0$.
\end{thm}

This theorem extends directly to the ring of Laurent polynomials.
\begin{lem}
McCoy's theorem holds for the Laurent polynomial ring $S[x_1^{\pm}, \dots, x_m^{\pm}]$.
\label{lem:mccoy_laurent}
\end{lem}
\begin{proof}
Suppose $\alpha, \beta \in S[x_1^{\pm}, \dots, x_m^{\pm}]$ are non-zero polynomials such that $\alpha\beta = 0$. We can choose a monomial $m = x_1^{k_1} \cdots x_m^{k_m}$ such that $m\alpha$ and $m\beta$ are ordinary polynomials in $S[x_1, \dots, x_m]$. Their product is $(m\alpha)(m\beta) = m^2(\alpha\beta) = 0$. By the standard McCoy's theorem, there exists a non-zero $s \in S$ such that $s(m\alpha) = 0$. Since $m$ is a unit in the Laurent polynomial ring, we can multiply by its inverse $m^{-1}$ to find $s\alpha=0$.
\end{proof}

With the help of McCoy's theorem, we show the following property of topological BB codes.
\begin{lem}
If the $\mathbf{BB}_{p^k}(f,g)$ code is topological, then $p \nmid f$ or $p \nmid g$ in the ring $R = \mathbb{Z}_{p^k}[x^\pm, y^\pm]$.
\label{lem:TO_and_GSD}
\end{lem}
\begin{proof}
Assume for contradiction that $p \mid f$ and $p \mid g$. Let $p^t$ be the highest power of $p$ (with $1 \le t < k$) that divides both $f$ and $g$, so we can write $f = p^t f_0$ and $g = p^t g_0$, where at least one of $f_0, g_0$ is not divisible by $p$. By Lemma~\ref{lem:mccoy_laurent}, this means at least one of $f_0, g_0$ is not a zero divisor in $R$.

The topological condition implies that the chain complex
\begin{equation}
    R\xrightarrow{{\partial_B = {g \choose -f} }}{R^2}\xrightarrow{{\varepsilon_A = (f\; g)}}R \label{eq:exact_proof_0}
\end{equation}
is exact. Consider the element $(g_0, -f_0)^T \in R^2$. It lies in the kernel of $\varepsilon_A$ because $\varepsilon_A (g_0,-f_0)^T = f g_0 - g f_0 = (p^t f_0)g_0 - (p^t g_0)f_0 = 0$.
Since the complex is exact, $\ker\varepsilon_A = \operatorname{im}\partial_B$, so there must exist some $h \in R$ such that:
\begin{equation}
    \begin{pmatrix} g_0 \\ -f_0 \end{pmatrix} = h \begin{pmatrix} g \\ -f \end{pmatrix} = h p^t \begin{pmatrix} g_0 \\ -f_0 \end{pmatrix}. \label{eq:exact_proof}
\end{equation}
Since at least one of $f_0$ or $g_0$ is not a zero divisor, we can conclude that $1 = h p^t$. This is impossible, as $p^t$ is not a unit in $R = \mathbb{Z}_{p^k}[x^\pm, y^\pm]$ for $t \ge 1$. Thus, our initial assumption must be false.
\end{proof}

\subsection{Proof of Theorem~\ref{lem:topo_prime_factors}}

With these preliminaries, we now prove Theorem~\ref{lem:topo_prime_factors}.

\begin{proof}[Proof of Theorem~\ref{lem:topo_prime_factors}]
By the Chinese remainder theorem, it suffices to prove the statement for $N=p^k$. We must show that the complex in Eq.~\eqref{eq:exact_proof_0} over $R = \mathbb{Z}_{p^k}[x^\pm, y^\pm]$ is exact if and only if the corresponding complex over the field $R_p = \mathbb{Z}_p[x^\pm, y^\pm]$ is exact.

($\implies$) Assume the complex over $R$ is exact. To show exactness over $R_p$, we need to prove $\ker\varepsilon_{Ap} \subseteq \operatorname{im}\partial_{Bp}$. Let $(r_p, s_p)^T \in \ker\varepsilon_{Ap}$, and let $r, s \in R$ be lifts of these elements. The condition $r_p f_p + s_p g_p = 0$ implies $rf+sg \in pR$. Multiplying by $p^{k-1}$ gives $p^{k-1}(rf+sg) = 0$, so $p^{k-1}(r,s)^T \in \ker\varepsilon_A$. By exactness over $R$, there exists $h \in R$ such that
\begin{equation}
    p^{k-1}\begin{pmatrix} r \\ s \end{pmatrix} = h \begin{pmatrix} g \\ -f \end{pmatrix}. \label{eq:relation_topo_2}
\end{equation}
Multiplying by $p$ yields $0 = ph(g, -f)^T$. By Lemma~\ref{lem:TO_and_GSD}, at least one of $f$ or $g$ is not a zero divisor, which implies $ph=0$. This means $h$ must be a multiple of $p^{k-1}$, so we can write $h=p^{k-1}h_0$ for some $h_0 \in R$. Substituting this into Eq.~\eqref{eq:relation_topo_2} and canceling the common factor of $p^{k-1}$ gives $(r,s)^T = h_0(g,-f)^T+pR^2$. Reducing this equation modulo $p$, we get $(r_p, s_p)^T = h_{0p}(g_p, -f_p)^T$, which shows that $(r_p, s_p)^T \in \operatorname{im}\partial_{Bp}$. Thus, the complex over $R_p$ is exact.

($\impliedby$) Assume the complex over $R_p$ is exact. We need to show $\ker\varepsilon_A \subseteq \operatorname{im}\partial_B$. Let $(r,s)^T \in \ker\varepsilon_A$. Reducing modulo $p$, $(r_p, s_p)^T$ is in $\ker\varepsilon_{Ap}$. By exactness over $R_p$, there exists $h^{(0)} \in R$ such that $(r_p, s_p)^T = h^{(0)}_p(g_p, -f_p)^T$. This implies that the difference is a multiple of $p$:
\begin{equation}
    \begin{pmatrix} r \\ s \end{pmatrix} - h^{(0)}\begin{pmatrix} g \\ -f \end{pmatrix} = p \begin{pmatrix} r^{(1)} \\ s^{(1)} \end{pmatrix} \quad \text{for some } r^{(1)}, s^{(1)} \in R. \label{eq:relation_topo_3}
\end{equation}
Applying $\varepsilon_A$ to both sides, and noting that the left side is in $\ker\varepsilon_A$, we find $p \cdot \varepsilon_A (r^{(1)}, s^{(1)})^T = 0$. Since $p$ is a zero-divisor, this implies $\varepsilon_A (r^{(1)}, s^{(1)})^T$ is a multiple of $p^{k-1}$, which is sufficient to show $(r^{(1)}_p, s^{(1)}_p)^T \in \ker\varepsilon_{Ap}$.
We can therefore repeat this process: find $h^{(1)}$ such that $(r^{(1)}, s^{(1)})^T - h^{(1)}(g, -f)^T$ is a multiple of $p$. Iterating this procedure $k$ times, we construct $h = h^{(0)} + p h^{(1)} + \dots + p^{k-1}h^{(k-1)}$ such that
\begin{equation}
    \begin{pmatrix} r \\ s \end{pmatrix} - h \begin{pmatrix} g \\ -f \end{pmatrix} = p^k \begin{pmatrix} r^{(k)} \\ s^{(k)} \end{pmatrix} = 0 \quad \text{in } R.
\end{equation}
This shows that $(r,s)^T \in \operatorname{im}\partial_B$, completing the proof.
\end{proof}

\subsection{Proof of Corollary~\ref{cor:topo_finite}}

Corollary~\ref{cor:topo_finite} claims that a $\mathbf{BB}_N(f,g)$ code is topological if and only if the quotient ring $R/(f,g)$ has finite cardinality. By Theorem~\ref{lem:topo_prime_factors}, one only needs to show this for $\mathbf{BB}_p(f,g)$ codes.

\begin{proof}[Proof of Corollary~\ref{cor:topo_finite}]
By Theorem~\ref{lem:topo_prime_factors}, the code is topological if and only if the $\mathbf{BB}_p(f,g)$ code is topological for every prime factor $p$ of $N$. By a standard result (a consequence of Nakayama's Lemma)~\cite{Eisenbud1995}, the $\mathbb{Z}_N$-module $R/(f,g)$ is finite if and only if the $\mathbb{Z}_p$-vector space $R_p/(f_p, g_p)$ is finite-dimensional for every prime $p\mid N$. Thus, it suffices to prove the corollary for the case $N=p$ is a prime.

For $R_p = \mathbb{Z}_p[x^{\pm}, y^{\pm}]$, the code is topological if and only if the chain complex
\begin{equation}
    R_p \xrightarrow{\partial_{Bp}} R_p^2 \xrightarrow{\varepsilon_{Ap}} R_p \label{eq:chain_complex_ptB}
\end{equation}
is exact, so the chain complex $0\to R_p \xrightarrow{\partial_{Bp}} R_p^2 \xrightarrow{\varepsilon_{Ap}} R_p$ (the injectivity of $\partial_{Bp}$ is trivial, as in Lemma~\ref{lem:TO_and_GSD}).
We invoke the Buchsbaum-Eisenbud exactness criterion~\cite{BUCHSBAUM1973} for a complex of free modules over the ring $R_p$. The criterion requires, among other things, that the codimension (or height) of the ideal of maximal minors of $\varepsilon_{Ap}$ be at least 2. The matrix for $\varepsilon_{Ap}$ is $(f_p, g_p)$, so its ideal of maximal ($1 \times 1$) minors is simply the ideal $I = (f_p, g_p)$.
The ring $R_p$ has Krull dimension 2. Therefore, the condition becomes $\operatorname{codim}(f_p, g_p) \ge 2$. Since codimension cannot exceed the dimension of the ring, this forces $\operatorname{codim}(f_p, g_p) = 2$.

By a fundamental result in commutative algebra (unmixedness theorem)~\cite{Eisenbud1995}, for an ideal $I$ in a polynomial ring over a field, its codimension and the Krull dimension of the quotient ring are related by:
\begin{equation}
    \dim(R_p/I) + \operatorname{codim}(I) = \dim(R_p).
\end{equation}
Substituting our values, we find that $\operatorname{codim}(f_p, g_p)=2$ is equivalent to $\dim(R_p/(f_p, g_p)) = 0$. For a finitely generated algebra over a field like $\mathbb{Z}_p$, having Krull dimension 0 is equivalent to being a finite-dimensional vector space. A finite-dimensional vector space over a finite field has finite cardinality.

Thus, for each prime $p$, the code being topological is equivalent to $|R_p/(f_p,g_p)|$ being finite. This completes the proof.
\end{proof}

\section{Legitimate excitation patterns}
\label{app2}
In this appendix, we show that every element of $R^2$ is a legitimate excitation pattern, meaning that it can be created by either a local or an infinitely extended operator. By linearity, it is sufficient to prove this for the generators $(1,0)^T$ and $(0,1)^T$.

Since the argument for each generator is analogous, our focus reduces to $(1,0)^T$. The condition for it to be legitimate is given by the following lemma.

\begin{lem}
    For a $\mathbf{BB}_N(f,g)$ code, the vector $(1,0)^T$ corresponds to a legitimate excitation if and only if $p\nmid f$ or $p\nmid g$.\label{lem:local_operator_can_be_created}
\end{lem}

\begin{proof}
($\implies$) Assume that $\left(1,0\right)^{T}$ represents a legitimate
excitation. By definition, this means there are formal Laurent series $\alpha$ and $\beta$ (in variables $x$ and $y$) such that
\begin{equation}
1=\alpha f+\beta g.
\end{equation}
(A finite series, i.e., a polynomial, corresponds to a local operator, while an infinite series corresponds to an infinitely extended operator.) 

We proceed by contradiction. Assume that the conclusion is false, i.e.,
that $p\mid f$ and $p\mid g$. Then $p$  also divides $\alpha f+\beta g$.
This implies that $p$ divides $1$. This is a contradiction, as $p$ is not a unit in $\mathbb{Z}_{p^k}$.

($\impliedby$) Without loss of generality, assume $p\nmid f$. We
will show that $\left(1,0\right)^{T}$ is a legitimate excitation
by constructing a formal Laurent series $\alpha$ such that $\alpha f=1$.

First, we can use the division algorithm to write $f=pq+r$, where
$q$ and $r$ are Laurent polynomials. Since $p\nmid f$, the remainder
$r$ is non-zero. Since $p^{k}=0$ in the ring $\mathbb{Z}_{p^{k}}$,
expanding $r^{k}=\left(f-pq\right)^{k}$ shows that $r^{k}$ is divisible
by $f$. This means there exists a Laurent polynomial $\xi\in\mathbb{Z}_{p^{k}}\left[x^{\pm},y^{\pm}\right]$
such that 
\begin{equation}
r^{k}=f\cdot\xi.\label{eq:rk}
\end{equation}

Next, according to Fact~\ref{fact1} (given below), it is possible to choose a weighted
degree such that the lowest-degree term in $r$ is unique; let us
call it $cx^{a_{1}}y^{a_{2}}$. The lowest-degree term of $r^{k}$
is then $c^{k}x^{ka_{1}}y^{ka_{2}}$. By construction, $p\nmid c$,
so $c$ is invertible in $\mathbb{Z}_{p^{k}}$. 
Dividing both sides of Eq.~\eqref{eq:rk} by $c^{k}x^{ka_{1}}y^{ka_{2}}$ gives
\begin{equation}
1-u=f\cdot h, \label{eq:fh}
\end{equation}
where $u$ and $h$ denote polynomials defined as
\begin{equation}
u\coloneqq1-c^{-k}x^{-ka_{1}}y^{-ka_{2}}\cdot r^{k},\qquad h\coloneqq c^{-k}x^{-ka_{1}}y^{-ka_{2}}\xi.
\end{equation}
By construction, each term in $u$ has a strictly positive weighted
degree. Thus, 
$
\alpha\coloneqq h\sum_{i=0}^{\infty}u^{i}    
$
is a well-defined formal Laurent series, and Eq.~\eqref{eq:fh} implies $\alpha f=1$. Taking $\beta=0$, we have $\alpha f+\beta g=1$, which proves that $\left(1,0\right)^{T}$
is a legitimate excitation.
\end{proof}

The preceding proof has used the following fact.

\begin{fact} \label{fact1}
For any non-zero Laurent polynomial $f$ in $d$ variables, a weight vector $\mathbf{w}=(w_1, \dots, w_d)$ with positive integer components can be chosen to define a weighted degree under which $f$ has a unique term of minimal degree.
\end{fact}

\begin{proof}
Without loss of generality, we assume $f$ is an ordinary polynomial,
since it can be made so by multiplication with an appropriate monomial. 

An explicit choice of $\mathbf{w}$ can be made in the following steps.

1) \textit{Identify exponents.}
Terms in $f$ look like $c \cdot x_1^{a_1} \cdots x_d^{a_d}$. We represent each term by its exponent vector $\mathbf{a}=(a_1, \dots, a_d)$. Each component is non-negative, since $f$ is an ordinary polynomial. Let $S$ be the set of all such exponent vectors for the terms appearing in $f$.

2) \textit{Select a unique candidate.}
We use lexicographical order to find a unique ``smallest'' vector in $S$. A vector $\mathbf{a}$ is lexicographically smaller than $\mathbf{b}$ if the first non-zero entry in the difference vector $\mathbf{b}-\mathbf{a}$ is positive. Since $S$ is a finite, non-empty set of vectors, it has a unique lexicographically minimal element. Let us call it $\mathbf{a}_0$.

3) \textit{Construct the weights.}
Now we define a weight vector $\mathbf{w}$ that will force the term corresponding to $\mathbf{a}_0$ to have the unique lowest degree. Let $A_{\max}$ be the maximum value of any single exponent in any vector in $S$. For example, if $f=x^5 y^2 + z^{12}$, then $A_{\max}=12$. Choose an integer $M > A_{\max}$. Define the weight vector as a sequence of powers of $M$:
\begin{equation}
    \mathbf{w} = (M^{d-1}, M^{d-2}, \dots, M^1, M^0).
\end{equation}

4) \textit{Verify uniqueness.}
The weighted degree of a term with exponent vector $\mathbf{a}$ is the dot product $\mathbf{w} \cdot \mathbf{a}$. This is essentially the value of the number whose digits are $(a_1, \dots, a_d)$ in base $M$. Namely,
\begin{equation}
    \text{deg}_{\mathbf{w}}(x^\mathbf{a}) = \sum_{i=1}^d a_i M^{d-i}
\end{equation}
Since we chose $M$ to be larger than any possible ``digit'' $a_i$, comparing these degree values is exactly equivalent to comparing the exponent vectors lexicographically.

Because $\mathbf{a}_0$ is the unique lexicographically smallest vector in $S$, its corresponding weighted degree $\mathbf{w} \cdot \mathbf{a}_0$ will be uniquely minimal. 

Thus, we have constructed a weighted degree such that the lowest-degree term in $f$ is unique.
\end{proof}

According to Lemma~\ref{lem:TO_and_GSD}, for a topological $\mathbb{Z}_{p^k}$ BB code, the Laurent polynomials $f$ and $g$ that specify its stabilizer map cannot be zero divisors simultaneously. Thus, the condition in Lemma~\ref{lem:local_operator_can_be_created} generically holds, indicating that each element in $R^2$ describes a legitimate excitation pattern. 

This can be generalized to arbitrary composite qudit dimensions by the Chinese remainder theorem, where the legitimate excitation condition in Lemma~\ref{lem:local_operator_can_be_created} becomes that the greatest common divisor of all the coefficients of $f$ and $g$ must be coprime to qudit dimension $N$. Therefore, any element in $R^2$ corresponds to a legitimate excitation pattern for its topological BB code, and its anyon types are labeled by $R^2/\im\pt$.

\section{Proof of Theorem~\ref{thm:anyon_type_relation}\label{app2.5}}

In this appendix, we provide a proof of Theorem~\ref{thm:anyon_type_relation}.

We will make use of a homological algebra tool known as the torsion functor over modules, or $\operatorname{Tor}$. We consider a commutative ring $S$ and two $S$-modules $M$ and $N$. A free resolution of $N$ is constructed as the following exact sequence:
\begin{equation}
    \dots \to F_{i+1}\xrightarrow{\varphi_{i+1}}F_i\xrightarrow{\varphi_{i}}F_{i-1}\to\dots\to F_0\xrightarrow{\varphi_0} N\to 0,\label{eq:free_res}
\end{equation}
where $F_i$ denotes free $S$-modules, whereas $\varphi_i$ represents a module homomorphism from $F_{i}$ to $F_{i-1}$ for each $i$. Tensoring the free resolution Eq.~\eqref{eq:free_res} with $M$ induces the following chain complex: 
\begin{equation}
    \dots \to M\otimes_S F_{i+1}\xrightarrow{1_M\otimes_S\varphi_{i+1}}M\otimes_S F_i\to\dots\to M\otimes_
    SF_0\to 0,\label{eq:free_res2}
\end{equation}
where $1_M$ denotes the identity map on $M$. Then, the $i$-th Tor group is defined as the $i$-th homology group of the chain complex Eq.~\eqref{eq:free_res2}. Concretely, 
\begin{equation}
    \operatorname{Tor}_{i}^S(N,M)\coloneqq \begin{cases}
        M\otimes_S N,\quad i=0.\\    \ker(1_M\otimes_S\varphi_{i})/\im(1_M\otimes_S\varphi_{i+1}),\quad i\geq 1.
    \end{cases}
\end{equation}
  
We remark that $\operatorname{Tor}_{i}^S(N,M)=\operatorname{Tor}_{i}^S(M,N)$, and that the Tor groups do not depend on the particular free resolution chosen in Eq.~\eqref{eq:free_res}. For further details on torsion functors, see standard references such as~\cite{Rotman2009, Eisenbud1995, StacksProject}. The result needed for our purposes is formulated in the following lemma:
\begin{lem}[\href{https://stacks.math.columbia.edu/tag/051H}{Tag 051H} of Ref.~\cite{StacksProject}]
    Let $S$ be a commutative ring, $I\subset S$ an ideal, $M$ an $S$-module with a family of elements $\{m_{\alpha}\}_{\alpha \in A}$ indexed by a set $A$, and $[m_{\alpha}] \coloneqq m_{\alpha}+IM$. Assume
    \begin{enumerate}
    \item[(i)] $I$ is nilpotent,
    \item[(ii)] $\{[m_{\alpha}]\}_{\alpha\in A}$  forms a basis of $M/IM$ over $S/I$, and
    \item[(iii)] $\operatorname{Tor}_1^S(S/I,M)=0$.
    \end{enumerate}
Then, $M$ is a free $S$-module with basis $\{m_{\alpha}\}_{\alpha\in A}$.\label{lem:tor1}
\end{lem}

A proof can be found in Ref.~\cite[\href{https://stacks.math.columbia.edu/tag/051H}{Tag 051H}]{StacksProject}. To compute the anyon types and their fusion rules, we consider $\mathcal{Q} = R/(f,g)$ for the $\mathbf{BB}_{p^k}(f,g)$ code.
The following lemma shows that the $\operatorname{Tor}_1=0$ condition holds for topological $\mathbb{Z}_{p^k}$ BB codes.

\begin{lem}
Let $\mathcal{Q} = R/(f,g)$ be the module associated with the topological $\mathbf{BB}_{p^k}(f,g)$ code. Then
\begin{equation}
    \operatorname{Tor}_1^{\mathbb{Z}_{p^k}}(\mathbb{Z}_{p^k}/(p),\mathcal{Q})=0.
\end{equation}\label{lem:tor2}
\end{lem}
\begin{proof}
    We construct the following free resolution
    \begin{equation}
        \dots\xrightarrow{\times p}\mathbb{Z}_{p^k}\xrightarrow{\times p^{k-1}} \mathbb{Z}_{p^k}\xrightarrow{\times p} \mathbb{Z}_{p^k}\xrightarrow{\pi} \mathbb{Z}_{p^k}/(p)\to 0.
    \end{equation}
    Here $\times p$ and $\times p^{k-1}$ represent the multiplication by $p$ and $p^{k-1}$, respectively. The canonical projection is denoted by $\pi$. Tensoring the free resolution with $\mathcal{Q}$ results in the chain complex
    \begin{equation}
        \dots\xrightarrow{\times p}\mathcal{Q}\xrightarrow{\times p^{k-1}} \mathcal{Q}\xrightarrow{\times p} \mathcal{Q}\xrightarrow{\pi} \mathcal{Q}/p\mathcal{Q}\to 0.
    \end{equation}
    Then by definition, we have
    \begin{equation}
        \operatorname{Tor}_1^{\mathbb{Z}_{p^k}}(\mathbb{Z}_{p^k}/(p),\mathcal{Q})=\frac{\ker(\times p)}{\im (\times p^{k-1})}=\frac{\operatorname{Ann}_{\mathcal{Q}}(p)}{p^{k-1}\mathcal{Q}}.\label{eq:tor1}
    \end{equation}
    Here, $\operatorname{Ann}_{\mathcal{Q}}(p)\coloneqq\{q\in \mathcal{Q}\,|\,pq=0\}$ is the annihilator of $p$ in $\mathcal{Q}$.
      
    To prove the lemma, we will show that $\operatorname{Ann}_{\mathcal{Q}}(p) = p^{k-1}\mathcal{Q}$. The inclusion ${p^{k-1}\mathcal{Q}}\subseteq{\operatorname{Ann}_{\mathcal{Q}}(p)}$ is immediate. We now prove the reverse inclusion, ${\operatorname{Ann}_{\mathcal{Q}}(p)}\subseteq{p^{k-1}\mathcal{Q}}$. Let $[m]\in\operatorname{Ann}_{\mathcal{Q}}(p)$ for some representative $m\in R$. By definition, $p[m]=0$ in $\mathcal{Q}$, which means $pm$ lies in the ideal $(f,g)_R$. Hence, there exist $r,s\in R$, such that
    \begin{equation}
         pm=rf+sg.\label{eq:relation_proof_1p}
    \end{equation}
    Multiplying $p^{k-1}$ on both sides of Eq.~\eqref{eq:relation_proof_1p} yields $p^{k-1}(rf+sg)=0$. This implies $p^{k-1}(r,s)^T\in\ker\eps_{A}$. By the exactness of Eq.~\eqref{eq:exact_proof_0}, there exists $h\in R$, such that $p^{k-1}(r,s)^T=h(g,-f)^T$. By Lemma~\ref{lem:TO_and_GSD}, at least one of $f$ or $g$ is not a zero divisor, which implies $ph=0$. This means $h$ must be a multiple of $p^{k-1}$, so we can write $h=p^{k-1}h_0$ for some $h_0 \in R$. Substituting this into Eq.~\eqref{eq:relation_topo_2} and canceling the common factor of $p^{k-1}$ gives $(r,s)^T=h_0(g,-f)^T+pR^2$. In other words, there exists $\alpha,\beta\in R$, such that
    \begin{equation}
        r-h_0g=p\alpha,\quad s+h_0f=p\beta.\label{eq:relation_proof_4p}
    \end{equation}
    Substituting Eq.~\eqref{eq:relation_proof_4p} into Eq.~\eqref{eq:relation_proof_1p} yields $p(m-\alpha f-\beta g)=0$. Consequently, $m'=m-\alpha f-\beta g$ is a multiple of $p^{k-1}$. In other words, $m'\in p^{k-1} R$. Projecting this relation to $\mathcal{Q}$, gives $[m]=[m']\in p^{k-1}\mathcal{Q}$. Therefore, we have that $p^{k-1}\mathcal{Q}=\operatorname{Ann}_{\mathcal{Q}}(p)$, and
    \begin{equation}
        \operatorname{Tor}_1^{\mathbb{Z}_{p^k}}(\mathbb{Z}_{p^k}/(p),\mathcal{Q})=\frac{\operatorname{Ann}_{\mathcal{Q}}(p)}{p^{k-1}\mathcal{Q}}=0.\label{eq:tor2}
    \end{equation}
\end{proof}
Then, we can determine the fusion structure of the BB code with prime power qudit dimensions by the following corollary:
\begin{prop}
    The fusion rules of anyons in the $\mathbf{BB}_{p^k}(f,g)$ code, which correspond to the additive group structure of $\mathcal{C}$, is given by
    \begin{equation}
        \mathcal{C}\cong\mathbb{Z}_{p^k}^{2Q},
    \end{equation}
    where $Q=\dim_{\mathbb{Z}_p}\mathcal{Q}_p$ with $\mathcal{Q}_p=\mathcal{Q}/p\mathcal{Q}$.\label{cor:zpk_fusion}
\end{prop}
\begin{proof}
    This corollary is proved by directly applying Lemma~\ref{lem:tor1}, where we take $S\coloneqq\mathbb{Z}_{p^k}$, $I\coloneqq(p)$, and $M\coloneqq\mathcal{Q}$. By definition, we have $M/IM=\mathcal{Q}_p$ and $S/I=\mathbb{Z}_p$. In this scenario, condition (i) is evident, condition (ii) holds since $\mathbb{Z}_p$ is a field, and condition (iii) is valid according to Lemma~\ref{lem:tor2}. Thus, Proposition~\ref{cor:zpk_fusion} follows since $\mathcal{C}=\mathcal{Q}\oplus\overline{\mathcal{Q}}$.
\end{proof}

Then, Theorem~\ref{thm:anyon_type_relation} is a natural generalization of Proposition~\ref{cor:zpk_fusion} to any composite qudit dimension $N$ cases by the Chinese remainder theorem.

\section{Derivation of Eq.~\eqref{eq:gsd1}} \label{app4}
Here we evaluate $\tr P$ to derive Eq.~\eqref{eq:gsd1}. 
  
To calculate the trace of the projector $P$, we first decompose it into a product of two mutually commuting projectors, $P = P_A P_B$, where
\begin{equation}
    P_{A}\coloneqq\frac{1}{\left|S_{A}\right|}\sum_{s_{A}\in S_{A}}s_{A} \quad \text{and} \quad P_{B}\coloneqq\frac{1}{\left|S_{B}\right|}\sum_{s_{B}\in S_{B}}s_{B}.
\end{equation}
We evaluate the trace in the basis $\{|\sigma\rangle\}$, defined by $|\sigma\rangle \coloneqq X_\sigma |00\cdots 0\rangle$ for $\sigma \in R_L^2$. In this expression, $X_\sigma$ is the tensor product of Pauli-$X$ matrices specified by $\sigma$. This gives 
\begin{equation}
    \tr P=\sum_{\sigma\in R_{L}^{2}}\left\langle \sigma\right|P_{A}P_{B}\left|\sigma\right\rangle.
\end{equation}
The derivation proceeds in several steps. First, the projector $P_B$ acts as the identity on states $|\sigma\rangle$ for which $\sigma\in \ker \varepsilon_{BL}$ and annihilates all others. The sum thus simplifies by restricting it to this kernel:
\begin{equation}
    \tr P=\sum_{\sigma\in\ker\varepsilon_{BL}}\left\langle \sigma\right|P_{A}\left|\sigma\right\rangle .
\end{equation}
Next, for any state $|\sigma\rangle$ remaining in the sum, the matrix element is constant: $\left\langle \sigma\right|P_{A}\left|\sigma\right\rangle =1/\left|S_{A}\right|$. The trace is thus the number of terms ($|\ker\varepsilon_{BL}|$) multiplied by this value. Using the identity $\left|S_{A}\right|=\left|\im \partial_{AL}\right|$, we obtain
\begin{equation}
    \tr P=\frac{\left|\ker\varepsilon_{BL}\right|}{\left|S_{A}\right|}=\frac{\left|\ker\varepsilon_{BL}\right|}{\left|\im \partial_{AL}\right|}.
\end{equation}
Furthermore, by the isomorphisms $\im \partial_{AL} \cong R_L / \ker \partial_{AL}$ and $\im \varepsilon_{BL} \cong R_L^2 / \ker \varepsilon_{BL}$, we have
\begin{equation}
\begin{aligned}
        \tr P &=\frac{\left|\ker \partial_{AL}\right|}{\left|R_{L}\right|}\frac{\left|R_{L}\right|^{2}}{\left|\im \varepsilon_{BL}\right|}=
        \left|\operatorname{Ann}_{R_L}(\bar{f},\bar{g})\right|\left|\frac{R_{L}}{(f,g)}\right|.\label{eq:computeP}
\end{aligned}
\end{equation}
  
In Eq.~\eqref{eq:computeP}, the second equality follows from definitions $\ker\partial_{AL}=\ker(\bar{f},\bar{g})^T=\operatorname{Ann}_{R_L}(\bar{f},\bar{g})$ and $\im\eps_{BL}=\im(g,-f)=(f,g)$.
  
Finally, we show $\left|\operatorname{Ann}_{R_L}(\bar{f},\bar{g})\right|=\left| R_L/(f,g)\right|$. To this end, we work in the standard basis $\{x^iy^j|i\in\mathbb{Z}_{L_x},j\in\mathbb{Z}_{L_y}\}$ of $R_L$. Consider the $\mathbb{Z}_N$-bilinear map
\begin{equation}
    B: R_L\times R_L\to \mathbb{Z}_N,\quad B(h_1,h_2)\coloneqq\langle h_1,h_2\rangle=(\bar{h}_1h_2)_0.\label{eq:bilinear0}
\end{equation}
Here, the $(\bar{h}_1h_2)_0$ represents the constant term of $\bar{h}_1h_2$. The orthogonal complement of a submodule $M$ of $R_L$ with respect to the bilinear map $B$ is defined as
\begin{equation}
    M^{\perp}=\{r\in R_L|\langle r,m\rangle=0,\forall m\in M\}.
\end{equation}
  
Viewing $R_L$ as a finite Abelian group, all linear maps $\chi(\bullet)=\langle h_1,\bullet\rangle$ in Eq.~\eqref{eq:bilinear0} form the \emph{character group} $\hat{R}_L=\hom(R_L,\mathbb{Z}_N)$ of $R_L$, where $\hom$ denotes the set of all group homomorphisms from $R_L$ to $\mathbb{Z}_N$. We note that $R_L\cong \hat{R}_L$. This duality implies that any subgroup $M$ of $R_L$ and its orthogonal complement satisfy (see Corollary~9.3 in \cite{Lang2002})
\begin{equation}
    |M||M^{\perp}|=|R_L|.\label{eq:duality}
\end{equation}
In the case where $N$ is prime, Eq.~\eqref{eq:duality} specializes to the vector space identity $\dim_{\mathbb{Z}_N}M+\dim_{\mathbb{Z}_N}M^{\perp}=\dim_{\mathbb{Z}_N}R_L$. For more details on character groups, see standard references such as ~\cite{Lang2002, Apostol1976, Terras1999, Wood1999}. 
  
Now, let $M=(f,g)$. By definition, every $r\in (f,g)^{\perp}$ satisfies $\langle af+bg, r\rangle=0$ for all $a,b\in R_L$. This implies that $(f,g)^{\perp}=\operatorname{Ann}_{R_L}(\bar{f},\bar{g})$.
Combining this with Eq.~\eqref{eq:duality} yields
\begin{equation}
    \left|\operatorname{Ann}_{R_L}(\bar{f},\bar{g})\right|=\left|(f,g)^{\perp}\right|=\frac{|R_L|}{|(f,g)|}=\left|\frac{R_L}{(f,g)}\right|.
\end{equation}
Thus, we recover the relation describing the ground state degeneracy of BB codes in Eq.~\eqref{eq:gsd1}.

\section{Topological index of BB codes with the toric layout}
\label{app3}
In this appendix, we explain how to evaluate the fusion rules of the BB code using the algebraic-geometric approach. Explicitly, we consider the $\mathbf{BB}_{p^k}(f,g)$ code with the toric layout. By Theorem~\ref{thm:anyon_type_relation}, this requires us to calculate the topological index of the $\mathbf{BB}_{p^k}(f,g)$ code through the BKK theorem.

\begin{figure*}[t]
    \centering
    \includegraphics[width=0.7\linewidth]{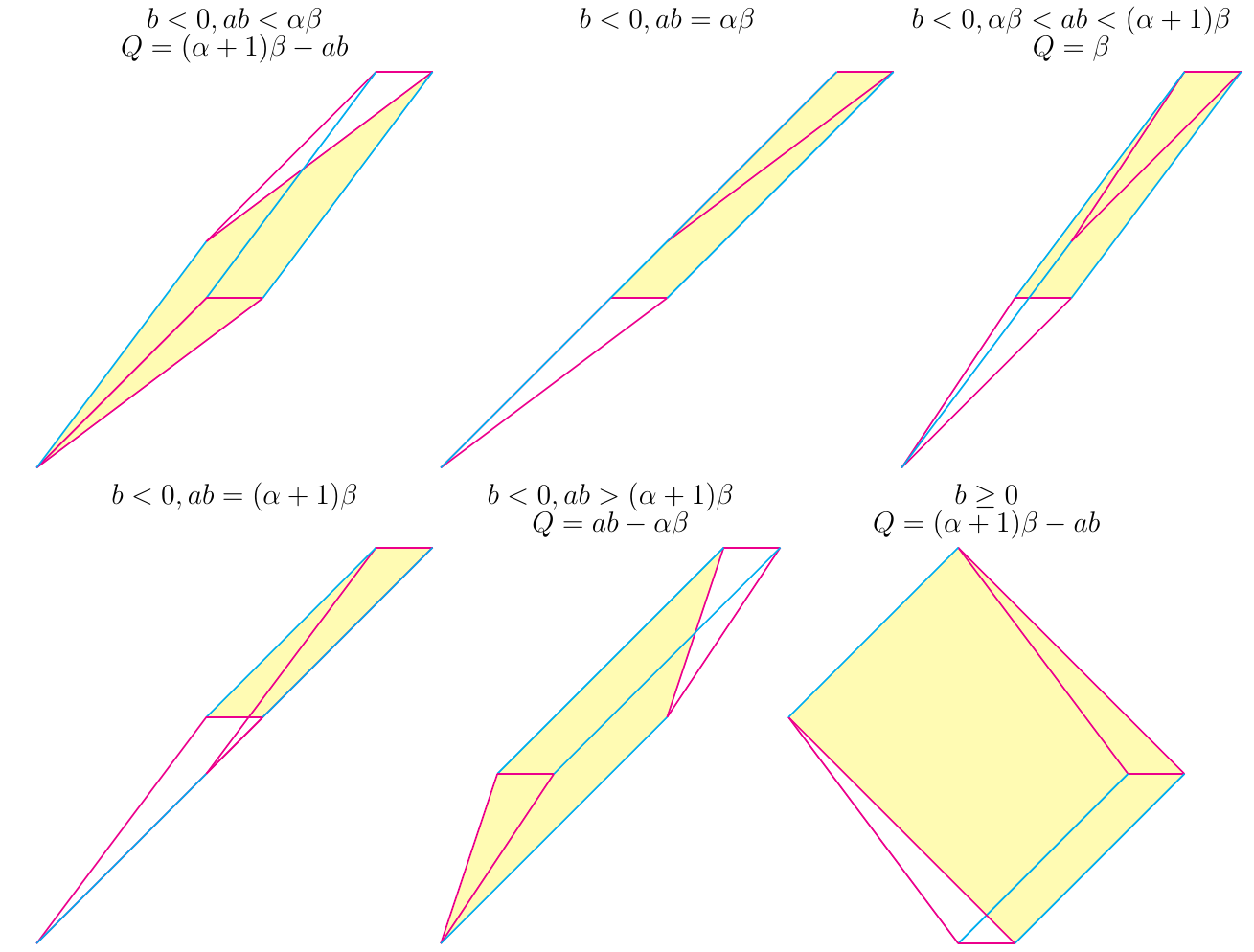}
    \caption{The topological indices of the $\mathbf{BB}_p(f,g)$ code with the toric layout for case 4 under different parameters. Throughout, we assume $\alpha,\beta\geq 0$ and $a\leq 0$. The pink triangles and cyan segments indicate the Newton polygons of $f_p$ and $g_p$, respectively, and the yellow shaded regions represent the mixed area in each case. The parameter ranges and corresponding topological indices for each case are displayed above the figures. For the boundary cases shown in the second figure on the top row and the first figure on the bottom row, the equality of the BKK theorem does not hold, and the topological indices depend explicitly on the coefficients of $f_p$ and $g_p$.}
    \label{fig:topological_index_4}
\end{figure*}

\subsection{Preliminaries}
\subsubsection{Origin shift and axis reflection operations}

We note that the topological indexes of $\mathbf{BB}_N(f,g)$ codes remain invariant under the following two type transformations.

First, the code is invariant under the transformation
\begin{equation}
    \mathbf{BB}_N(f,g)\rightarrow\mathbf{BB}_N(m_1f,m_2g),\label{eq:shift}
\end{equation}
where $m_1$ and $m_2$ are monomials. This is equivalent to shifting the origins of the coordinate systems labeling the qubits \cite{Chen2025}. 

Secondly, the topological index is also invariant under the transformations
\begin{equation}
    \begin{aligned}
        &\mathbf{BB}_N(f(x,y),g(x,y))\rightarrow\mathbf{BB}_N(f(\bar{x},y),g(\bar{x},y)),\\
        &\mathbf{BB}_N(f(x,y),g(x,y))\rightarrow\mathbf{BB}_N(f(x,\bar{y}),g(x,\bar{y})),\label{eq:reflection}
    \end{aligned}
\end{equation}
which correspond to axis reflections about the $x$- and $y$-axis, respectively. 

These two types of transformations are instrumental in simplifying the analysis that follows.

\subsubsection{BKK theorem}
We now introduce the BKK theorem, our main algebraic-geometric tool for evaluating the topological index. We focus on the bivariate polynomial case. To prepare for the theorem, we first define Newton polygons and the concept of mixed area. (The BKK method was applied to qubit codes in \cite{Chen2025} and is generalized here to arbitrary qudits.)

We use the shorthand $\mathcal{X}^{\mathbf{j}}$ for $x^{j_1}y^{j_2}$ with $\mathbf{j}=(j_1,j_2)$. The Newton polygon of a bivariate Laurent polynomial 
\begin{equation}
    h=\sum_{\mathbf{j}=(j_1,j_2)\in\mathbb{Z}^2}c_{\mathbf{j}}\mathcal{X}^{\mathbf{j}}
\end{equation}
is defined as the convex hull of the set of its exponent vectors $\{\mathbf{j}\in\mathbb{Z}^2|c_{\mathbf{j}}\neq 0\}$. For example, the Newton polygon of the polynomial $h=1+x^{-\alpha}y^{b}$ is the segment connecting $(0,0)$ and $(-\alpha,b)$, and the Newton polygon of the polynomial $h=1+y+x^{a}y^{-\beta}$ is the triangle with vertices at $(0,0)$, $(0,1)$, and $(a,-\beta)$.
  
For two Newton polygons $\mathscr{P}_1$ and $\mathscr{P}_2$, their Minkowski sum is defined as:
\begin{equation}
    \mathscr{P}=\mathscr{P}_1+\mathscr{P}_2=\{v_1+v_2|v_1\in\mathscr{P}_1,v_2\in\mathscr{P}_2\}.
\end{equation}
Note that $\mathscr{P}$ is also a convex set. Any face of $\mathscr{P}$ can be uniquely decomposed into the  Minkowski sum of faces of $\mathscr{P}_1$ and $\mathscr{P}_2$. Specifically, for any face $S$ of $\mathscr{P}$, there exist unique faces $S_1$ of $\mathscr{P}_1$ and $S_2$ of $\mathscr{P}_2$, such that $S=S_1+S_2$. Furthermore, the restricted polynomial $h_{1, S_1}$ of the faces $S_1$ is defined as the sum of terms $h_1$ whose exponent vectors lie in $S_1$, and likewise for $h_{2, S_2}$.
  
The mixed area of two Newton polygons $\mathscr{P}_1$ and $\mathscr{P}_2$ is defined as:
\begin{equation}
    \operatorname{MV}(\mathscr{P}_1,\mathscr{P}_2)=\operatorname{Area}(\mathscr{P})-\operatorname{Area}(\mathscr{P}_1)-\operatorname{Area}(\mathscr{P}_2),
\end{equation}
where $\operatorname{Area}$ denotes the area of a polygon. 
  
We now present the BKK theorem. We consider the special case for two bivariate Laurent polynomials. A comprehensive treatment of a generalized form for multiple multivariate polynomials and its proof can be found in Ref.~\cite{Gorantla2021}.
\begin{thm}[BKK Theorem]
    Let $\mathbb{K}$ be a field and $\overline{\mathbb{K}}$ be the algebraic closure of $\mathbb{K}$. Let $h_1$ and $h_2$ be two polynomials in the Laurent polynomial ring $\mathbb{K}[x^\pm,y^\pm]$. Let $\mathscr{P}_1$ and $\mathscr{P}_2$ be the Newton polygons of $h_1$ and $h_2$. If $\dim_{\mathbb{K}}\frac{\mathbb{K}[x^\pm,y^\pm]}{(h_1,h_2)}$ is finite, then
    \begin{equation}
        \dim_{\overline{\mathbb{K}}}\frac{\overline{\mathbb{K}}[x^\pm,y^\pm]}{(h_1,h_2)}=\dim_{\mathbb{K}}\frac{\mathbb{K}[x^\pm,y^\pm]}{(h_1,h_2)}\leq\operatorname{MV}(\mathscr{P}_1,\mathscr{P}_2).
    \end{equation}
    If $\operatorname{MV}(\mathscr{P}_1,\mathscr{P}_2)>0$, then the equality holds if and only if for all faces $S=S_1+S_2$ of $\mathscr{P}$, their restricted polynomials $h_{1, S_1}$ and $h_{2, S_2}$ share no common roots in $(\overline{\mathbb{K}}^\times)^2$.
\end{thm}
In our case, we take $\mathbb{K}$ to be the finite field $\mathbb{Z}_p$ for a prime $p$ and $\overline{\mathbb{K}}\coloneqq\overline{\mathbb{F}}_p$. For the $\mathbf{BB}_{p}(f,g)$ code, we will call $Q=\dim_{\mathbb{Z}_p}\frac{\mathbb{Z}_p[x^{\pm},y^{\pm}]}{(f,g)}$ the topological index. The fusion rules of this code are given by $\mathbb{Z}_p^{2Q}$.

\subsection{Topological index}

We consider the $\mathbf{BB}_{p^k}(f,g)$ code with the toric layout. Namely, the Laurent polynomials $f$ and $g$ take the form
\begin{equation}
    f=c^{(0)}+c^{(1)}x+c^{(2)}x^{\bar{\alpha}}y^b,\quad g=d^{(0)}+d^{(1)}y+d^{(2)}x^ay^{\bar{\beta}}
\end{equation}
where $\bar{\alpha}=-\alpha$ and $\bar{\beta}=-\beta$ following the convention used in Ref.~\cite{Chen2025}. 

The topological index is determined by considering the corresponding $\mathbf{BB}_p(f,g)$ code with
\begin{equation}
    f_p=c^{(0)}_p+c^{(1)}_px+c^{(2)}_px^{\bar{\alpha}}y^b,\quad g_p=d^{(0)}_p+d^{(1)}_py+d^{(2)}_px^ay^{\bar{\beta}}.
\end{equation}
where the subscript $p$ indicates that coefficients are taken modulo $p$. Since the coefficients of $f_p$ and $g_p$ can be zero in $\mathbb{Z}_p$, the polynomials may reduce to binomials or monomials. Therefore, to determine the topological index, we perform a case-by-case analysis based on the number of non-zero terms in $f_p$ and $g_p$.

\begin{figure*}[t]
    \centering
    \includegraphics[width=\linewidth]{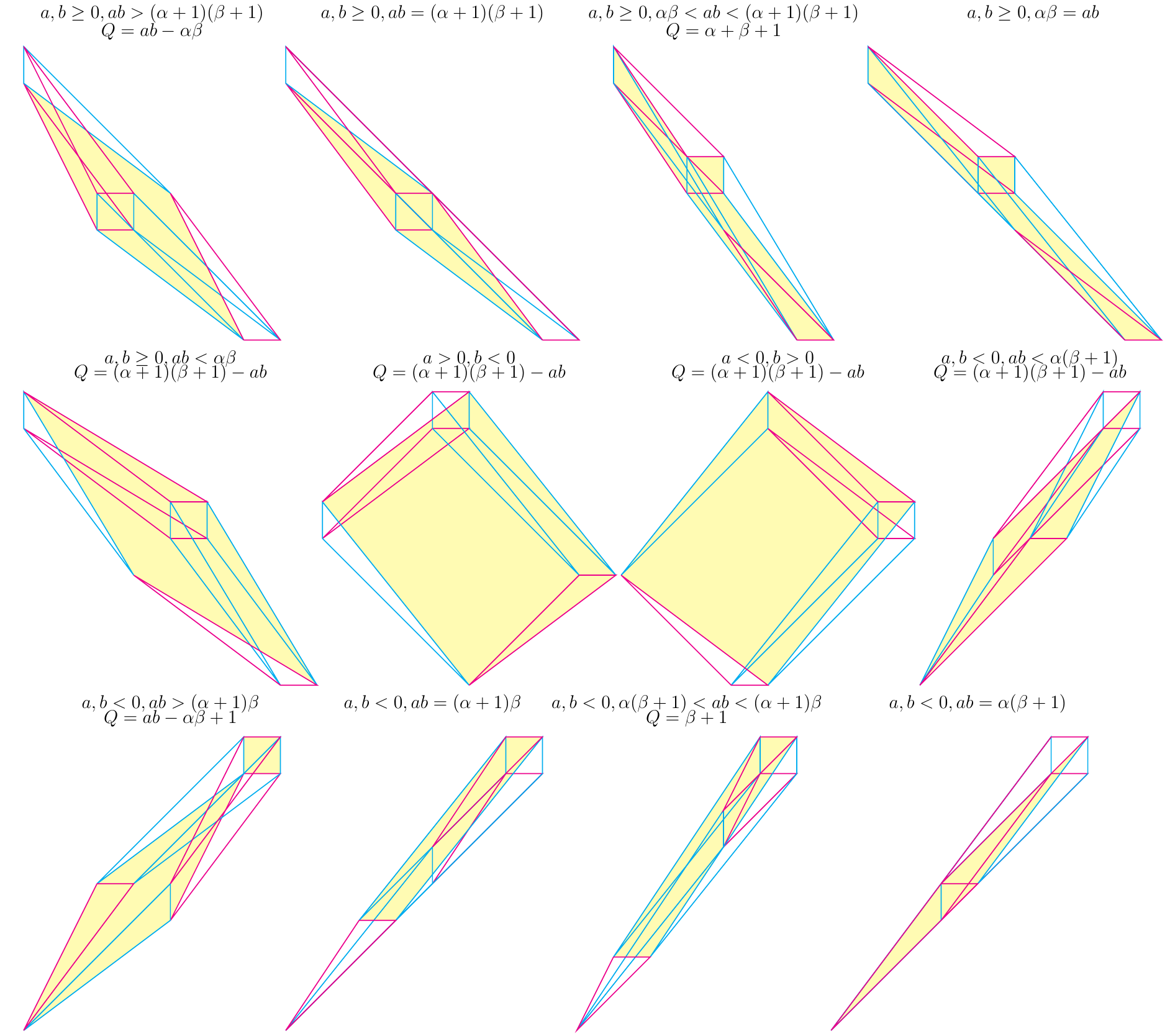}
    \caption{The topological indices of the $\mathbf{BB}_p(f,g)$ code with the toric layout for Case 5 with different parameters. In these cases, $\beta\geq\alpha\geq 0$. In each case, the pink and cyan triangles represent the Newton polygons of $f_p$ and $g_p$, respectively, while the mixed area is shown by the yellow shaded regions. The parameter ranges and corresponding topological indices for each case are displayed above the figures. For the boundary cases, the equality of the BKK theorem does not hold.}
    \label{fig:topological_index_5}
\end{figure*}

\textbf{Case 1: }Either $f_p$ or $g_p$ is a monomial. In this case, the quotient ring $R_p/(f_p,g_p)$ becomes trivial, so we can conclude $Q=0$.

\textbf{Case 2: }Either $f_p$ or $g_p$ reduces to $0$. In this case, if the other polynomial is a monomial, we have $Q=0$. Otherwise, the BB code is not topological.

\textbf{Case 3: }Both $f_p$ and $g_p$ are binomials. Without loss of generality, we assume 
\begin{equation}
    f_p=c^{(0)}_p+c^{(2)}_px^{\bar{\alpha}}y^b,\quad g_p=d^{(0)}_p+d^{(2)}_px^ay^{\bar{\beta}}.\label{eq:binomial_case}
\end{equation}
All other possibilities can be reduced to Eq.~\eqref{eq:binomial_case} by suitable shift transformations. For example, if $f_p=c^{(1)}_px+c^{(2)}_px^{\bar{\alpha}}y^b$ and $g_p=d^{(1)}_py+d^{(2)}_px^ay^{\bar{\beta}}$, then the code is transformed into the form of Eq.~\eqref{eq:binomial_case} by applying the transformation $f_p\to \bar{x}f_p$ and $g_p\to \bar{y}g_p$.

According to the finiteness criterion in Corollary~\ref{cor:topo_finite}, the code is \emph{not} topological if and only if two conditions are met simultaneously: 
\begin{equation}
    \alpha\beta-ab=0,\quad\left(-\frac{c^{(0)}_p}{c^{(2)}_p}\right)^{s}=\left(-\frac{d^{(0)}_p}{d^{(2)}_p}\right)^{t}.
\end{equation}
Here, $s=-a/\gcd(\alpha,a)$ and $ t=\alpha/\gcd(\alpha,a)$ arise from the parameterization $\alpha=-tq_1$, $\beta=-sq_2$, $a=sq_1$, and $b=tq_2$, with $q_1, q_2 \in \mathbb{Z}$ (If $\gcd(\alpha,a)=0$, then $s$ and $t$ are defined by $\beta$ and $b$). In all other scenarios, the code is topological. Specifically, if $\alpha\beta-ab=0$ but the equality is not met, the code is in a trivial phase with $Q=0$.

If the model is topological, the Newton polygons of $f_p$ and $g_p$ are two segments defined by vectors $(-\alpha,b)$ and $(a,-\beta)$. The mixed area is given by the area of the parallelogram spanned by these two vectors. According to the BKK Theorem, the topological index is
\begin{equation}
    Q=|\alpha\beta-ab|.
\end{equation}
  
\textbf{Case 4: } One polynomial is a trinomial and the other is a binomial. Without loss of generality, we focus on the case where $f_p$ is a trinomial and $g_p$ is a binomial. By a similar argument in Case 3, we can concentrate on the code determined by the following polynomials:
\begin{equation}
    f_p=c^{(0)}_p+c^{(1)}_px+c^{(2)}_px^{\bar{\alpha}}y^b,\quad  g_p=d^{(0)}_p+d^{(2)}_px^ay^{\bar{\beta}}.\label{eq:bi-trinomial_case}
\end{equation}
We restrict our analysis to the regime $\alpha \ge 0$, $\beta \ge 0$, and $a \le 0$, since negative parameters can be equivalently mapped to this domain by suitable shift and reflection transformations, namely $\alpha \to -\alpha + 1$, $\beta \to -\beta$, and $a \to -a$. By the finiteness criterion, the model is topological unless
\begin{equation}
    b=\beta=0,\quad\gcd(c^{(0)}_p+c^{(1)}_px+c^{(2)}_px^{\bar{\alpha}}, d^{(0)}_p+d^{(2)}_px^a)\ne 1.
\end{equation}

We obtain the topological index by the BKK theorem. The Newton polygons for $f_p$ and $g_p$ are a triangle and a segment, respectively. All possible Newton polygons and their corresponding topological indices are depicted in Fig.~\ref{fig:topological_index_4}.
  
\textbf{Case 5: } Both $f_p$ and $g_p$ are trinomials. In this case, the $\mathbb{Z}_p$ BB code is determined by
\begin{equation}
    f_p=c^{(0)}_p+c^{(1)}_px+c^{(2)}_px^{\bar{\alpha}}y^b,\quad g_p=d^{(0)}_p+d^{(1)}_py+d^{(2)}_px^ay^{\bar{\beta}},
\end{equation}
with coefficients being non-zero. It is sufficient to consider the scenario where $\alpha\geq 0$ and $\beta\geq 0$ after suitable shift and reflection transformations. Using the generalized Eisenstein’s criterion from Ref.~\cite{Chen2025}, the model is topological unless
\begin{equation}
    (\alpha,\beta,a,b)=(0,0,1,1),\quad \frac{c^{(0)}_p}{d^{(0)}_p}=\frac{c^{(2)}_p}{d^{(1)}_p}=\frac{c^{(1)}_p}{d^{(2)}_p}.
\end{equation}
The Newton polygons of $f_p$ and $g_p$ are two triangles. The topological indices, obtained via the BKK theorem, are illustrated in Fig.~\ref{fig:topological_index_5}.

\section{Sample codes in M2}
\label{app5}
In this appendix, we provide a brief demonstration of basic M2 command lines that are useful for computing the topological order and quasifractonic behavior in BB codes. We illustrate this with an example consistent with the main text. The code that produces Gr\"{o}bner basis is as follows:
\begin{lstlisting}[
      frame=TB,             
      framerule=0.4pt,
      captionpos=t,
      caption={Gr\"{o}bner basis computation.},
      label={code:groebner1},
      basicstyle={\ttfamily\fontsize{9pt}{11pt}\selectfont},
      breaklines=true,
      aboveskip=2pt,        
      belowskip=3pt
  ]
i1: S = ZZ[u, t, y, x, MonomialOrder => {Position => Up, Lex}]; --Define the polynomial ring.
i2: f = 1 + x + 6*t*y^5; g = 1 + y + 4*x^3*u; -- Define polynomials.
i4: I = ideal(12, t*x - 1, y*u - 1, f, g); --Define the ideal.
i5: gens gb I --Compute the Grobner Basis.

o5 = | 12 x-5 3y+3 y2+y-4 t-5 u-y-4 |
\end{lstlisting}

The command for generating a Gr\"{o}bner basis applies not only to rings, but also to modules, allowing us to directly check the topological condition as follows (For computational efficiency, we use the graded reverse lexicographic order here):
\begin{lstlisting} [
      frame=TB,             
      framerule=0.4pt,
      captionpos=t,
      caption={Topological condition check.},
      basicstyle={\ttfamily\fontsize{9pt}{11pt}\selectfont},
      breaklines=true,
      aboveskip=2pt,        
      belowskip=3pt
  ]
i1: S = ZZ[u, t, y, x]; --Define the polynomial ring.
i2: I'= ideal (12, x*t-1, y*u-1);
i3: Q'= S/I' -- Define quotient ring.
i4: f = 1 + x + 6*t*y^5; g = 1 + y + 4*x^3*u;
i6: eps = matrix{{f, g}}; partial = g || -f; --Define stabilizer and excitation maps.
i8: try (gens gb kernel eps - gens gb partial == 0) else false -- Check exactness. 

o8 = true
\end{lstlisting}

The periodicity can then be tested directly in the quotient ring $\mathcal{Q}$. For example, the periodicity along the $\hat{x}$ direction can be obtained with
\begin{lstlisting}[
      frame=TB,             
      framerule=0.4pt,
      captionpos=t,
      caption={Periodicity computation.},
      basicstyle={\ttfamily\fontsize{9pt}{11pt}\selectfont},
      breaklines=true,
      aboveskip=2pt,        
      belowskip=3pt
  ]
i1: S = ZZ[u, t, y, x];
i2: f = 1 + x + 6*t*y^5; g = 1 + y + 4*x^3*u; 
i4: I = ideal(12, t*x - 1, y*u - 1, f, g); 
i5: Q = S/I;
i6: i = 1; g = x; while i > 0 do (if g - 1 == 0 then break i; g = x*g, i = i + 1); i -- Check the ideal memberships.

o9 = 2
\end{lstlisting}
This gives that the periodicity along the $\hat{x}$ direction is 2. 

\twocolumngrid
\bibliographystyle{apsrev4-2}
\bibliography{CSSZn}

\begin{thebibliography}{123}
\bibitem [{not()}]{note}%
  \BibitemOpen
  \href@noop {} {}\bibinfo {note} {To avoid ambiguity in $ a^n $ for $ n < 0 $
  when defining exponential symmetries, we restrict $ a $ to be coprime with $
  N $ when discussing the gauging procedure. Elsewhere, $ a $ may be any
  element of $ \mathbb{Z}_N $.}\BibitemShut {Stop}%
\end{thebibliography}

\end{document}